\let\oldvec\vec
\documentclass{llncs}
\let\vec\oldvec

\usepackage{tikz}
\usetikzlibrary[positioning,decorations.pathmorphing,arrows]
\usepackage[normalem]{ulem}
\usepackage[english]{babel}

\usepackage{amsmath}
\usepackage{amsthm}
\usepackage{amssymb}
\usepackage{color}
\usepackage{mathrsfs}
\usepackage{prooftree}

\newcommand{\comment}[1]{}

\theoremstyle{definition}

\newcommand{\eg}{{\em e.g.~}}
\newcommand{\ie}{{\em i.e.~}}

\newcommand{\Coloneq}{\ensuremath{::=}}
\newcommand{\dbckslash}{\ensuremath{\setminus\!\!\setminus}}
\newcommand{\llbrace}{\{\!\!\{}
\newcommand{\rrbrace}{\}\!\!\}}
\newcommand{\llangle}{\langle\mkern-\thickmuskip\langle}
\newcommand{\rrangle}{\rangle\mkern-\thickmuskip\rangle}
\newcommand{\eqdef}{\mathrel{\raisebox{-1pt}{\ensuremath{\stackrel{\textit{\tiny{def}}}{=}}}}}
\newcommand{\eqalpha}{\ensuremath{=_{\alpha}}}
\newcommand{\emphdef}[1]{\textbf{\emph{#1}}}

\newcommand{\A}{\ensuremath{\mathcal{A}}}

\newcommand{\R}{\ensuremath{\mathcal{R}}}
\newcommand{\V}{\ensuremath{\mathcal{V}}}

\newcommand{\WN}[1]{\ensuremath{\mathcal{W\!N}_{#1}}}

\newcommand{\ctxt}[1]{\ensuremath{\mathtt{#1}}}

\newcommand{\pos}[1]{\ensuremath{\mathsf{#1}}}
\newcommand{\posset}[1]{\ensuremath{\mathcal{#1}}}
\newcommand{\mtp}[2]{\ensuremath{\mathsf{mtp}_{#2}({#1})}}

\newcommand{\Term}{\ensuremath{\mathcal{T}_{\mathtt{a}}}}

\newcommand{\TermExplicit}{\ensuremath{\mathcal{T}_{\mathtt{e}}}}

\newcommand{\TermVariable}{\ensuremath{\mathcal{X}}}

\newcommand{\TypeVariable}{\ensuremath{\mathcal{B}}}

\newcommand{\exsubs}[2]{\ensuremath{[{#1}\backslash{#2}]}}

\newcommand{\appterm}[2]{\ensuremath{{#1}\,{#2}}}
\newcommand{\absterm}[2]{\ensuremath{\lambda{#1}.{#2}}}
\newcommand{\substerm}[3]{\ensuremath{{#3}\exsubs{#1}{#2}}}

\newcommand{\valuetype}{\ensuremath{\mathtt{a}}}
\newcommand{\functtype}[2]{\ensuremath{{#1}\to{#2}}}
\newcommand{\intertype}[2]{\ensuremath{\multiset{#1}_{#2}}}
\newcommand{\M}{\ensuremath{\mathcal{M}}}
\newcommand{\N}{\ensuremath{\mathcal{N}}}

\newcommand{\sequ}[2]{\ensuremath{{#1}\vdash{#2}}}

\newcommand{\assign}[2]{\ensuremath{{#1}:{#2}}}
\newcommand{\derivable}[3]{\ensuremath{{#1}\rhd_{#3}{#2}}}

\newcommand{\ctxtapp}[2]{\ensuremath{{#1}\langle{#2}\rangle}}
\newcommand{\ctxtsum}[3]{\ensuremath{{#1}\mathrel{+_{#3}}{#2}}}
\newcommand{\ctxtres}[3]{\ensuremath{{#1}\mathrel{\dbckslash_{#3}}{#2}}}
\newcommand{\ctxtwoc}[2]{\ensuremath{{#1}\llangle{#2}\rrangle}}

\newcommand{\dom}[1]{\ensuremath{\mathtt{dom}({#1})}}

\newcommand{\applyax}[2]{\ensuremath{\mathtt{AX}({#1},{#2})}}
\newcommand{\applyval}[2]{\ensuremath{\mathtt{VAL}({#1},{#2})}}
\newcommand{\applyabs}[2]{\ensuremath{\mathtt{ABS}({#1},{#2})}}
\newcommand{\applyapp}[3]{\ensuremath{\mathtt{APP}({#1},{#2},{#3})}}

\newcommand{\projctxt}[1]{\ensuremath{\mathtt{CTXT}({#1})}}
\newcommand{\projprem}[1]{\ensuremath{\mathtt{PREM}({#1})}}
\newcommand{\projrule}[1]{\ensuremath{\mathtt{RULE}({#1})}}
\newcommand{\projsubj}[1]{\ensuremath{\mathtt{SUBJ}({#1})}}
\newcommand{\projtype}[1]{\ensuremath{\mathtt{TYPE}({#1})}}

\newcommand{\bv}[1]{\ensuremath{\mathtt{bv}({#1})}}
\newcommand{\fv}[1]{\ensuremath{\mathtt{fv}({#1})}}

\newcommand{\size}[1]{\ensuremath{\mathtt{sz}({#1})}}

\newcommand{\oc}[1]{\ensuremath{\mathtt{oc}({#1})}}
\newcommand{\roc}[1]{\ensuremath{\mathtt{r}\oc{#1}}}
\newcommand{\toc}[1]{\ensuremath{\mathtt{TOC}({#1})}}

\newcommand{\residuals}[2]{\ensuremath{{#1}/{#2}}}

\newcommand{\prefix}{\ensuremath{\leq}}

\newcommand{\rrule}[1]{\ensuremath{\mathrel{\mapsto_{#1}}}}
\newcommand{\rewrite}[1]{\ensuremath{\mathrel{\rightarrow_{#1}}}}
\newcommand{\rewriten}[1]{\ensuremath{\mathrel{\twoheadrightarrow_{#1}}}}

\newcommand{\reduction}[4]{\ensuremath{{#1}:{#2}\rewrite{#4}{#3}}}
\newcommand{\reductionn}[4]{\ensuremath{{#1}:{#2}\rewriten{#4}{#3}}}

\newcommand{\needed}{\ensuremath{\mathtt{nd}}}
\newcommand{\headnd}{\ensuremath{\mathtt{hnd}}}
\newcommand{\weaknd}{\ensuremath{\mathtt{whnd}}}

\newcommand{\callbyname}{\ensuremath{\mathtt{name}}}
\newcommand{\callbyneed}{\ensuremath{\mathtt{need}}}

\newcommand{\dB}{\ensuremath{\mathtt{dB}}}
\newcommand{\lsv}{\ensuremath{\mathtt{lsv}}}

\newcommand{\set}[1]{\ensuremath{\{{#1}\}}}
\newcommand{\multiset}[1]{\ensuremath{\llbrace{#1}\rrbrace}}

\newcommand{\substitute}[3]{\ensuremath{{{#3}\left\{{#1}\left/{#2}\right.\right\}}}}

\newcommand{\many}[2]{\ensuremath{({#1})_{#2}}}

\newcommand{\NF}[1]{\ensuremath{\mathcal{N\!F}_{#1}}}
\newcommand{\HNF}[1]{\ensuremath{\mathcal{H}\NF{#1}}}
\newcommand{\WHNF}[1]{\ensuremath{\mathcal{W}\HNF{#1}}}

\newcommand{\nf}[2]{\ensuremath{\mathtt{nf}_{#1}({#2})}}
\newcommand{\hnf}[2]{\ensuremath{\mathtt{hnf}_{#1}({#2})}}
\newcommand{\whnf}[2]{\ensuremath{\mathtt{whnf}_{#1}({#2})}}

\newcommand{\termat}[2]{\ensuremath{{#1}|_{#2}}}
\newcommand{\treeat}[2]{\ensuremath{{#1}|_{#2}}}
\newcommand{\replaceat}[3]{\ensuremath{{#1}[{#3}]_{#2}}}

\newcommand{\ruleName}[1]{\ensuremath{\mathtt{({#1})}}}
\newcommand{\ruleAxiom}{\ruleName{ax}}
\newcommand{\ruleArrowE}{\ruleName{\functtype{\!}{\!}e}}
\newcommand{\ruleArrowI}{\ruleName{\functtype{\!}{\!}i}}
\newcommand{\ruleValue}{\ruleName{val}}

\newcommand{\Rule}[3]{
    \prooftree
         {#1}
    \justifies  
         {#2}
    \thickness=0.05em
    \using
         {#3}
    \endprooftree}

\newcommand{\rootpos}{\ensuremath{\epsilon}}
\newcommand{\emptyred}{\ensuremath{\mathit{nil}}}
\newcommand{\obs}[1]{\cong_{#1}}

\title{Call-by-Need, Neededness and All That\thanks{This work was partially founded by LIA INFINIS.}}

\author{Delia Kesner\inst{1} \and Alejandro R\'{i}os\inst{2} \and Andr\'{e}s Viso\inst{3}}

\institute{
  IRIF, CNRS and Univ. Paris-Diderot
\and
  Universidad de Buenos Aires
\and
  CONICET and Universidad de Buenos Aires
}

\begin{document}

\maketitle

\begin{abstract}
We show that call-by-need is observationally equivalent to  weak-head needed
reduction. The proof of this result uses  a semantical argument based on a
(non-idempotent) intersection type system called $\V$. Interestingly, system
$\V$ also allows to syntactically identify all the weak-head needed redexes of
a term.
\end{abstract}

\section{Introduction}

One of the fundamental notions underlying this paper is the one of \emph{needed
reduction} in $\lambda$-calculus, which is to be used here to understand (lazy)
evaluation of functional programs. Key notions are those of reducible and
non-reducible programs: the former are programs (represented by
$\lambda$-terms) containing non-evaluated subprograms, called reducible
expressions (redexes), whereas the latter can be seen as definitive results of
computations, called normal forms. It turns out that every reducible program
contains a special kind of redex known as needed or, in other words, every
$\lambda$-term not in normal form contains a needed redex. A redex $\pos{r}$ is
said to be \emph{needed} in a $\lambda$-term $t$ if $\pos{r}$ has to be
contracted (\ie evaluated) sooner or later when reducing $t$ to
\emph{normal form}, or, informally said, if there is no way of avoiding
$\pos{r}$ to reach a normal form.

The needed strategy, which always contracts a needed redex, is
normalising~\cite{BarendregtKKS87}, \ie if a term can be reduced (in any way)
to a normal form, then contraction of needed redexes necessarily terminates.
This is an excellent starting point to design an evaluation strategy, but
unfortunately, neededness of a redex is not decidable~\cite{BarendregtKKS87}.
As a consequence, real implementations of functional languages cannot be
directly based on this notion.

Our goal is, however, to establish a clear connection between the semantical
notion of neededness and different implementations of lazy functional languages
(\eg Miranda or Haskell). Such implementations are based on \emph{call-by-need
calculi}, pioneered by Wadsworth~\cite{Wadsworth:thesis}, and extensively
studied \eg in~\cite{AriolaFMOW95}. Indeed, call-by-need calculi fill the gap
between the well-known operational semantics of the call-by-name
$\lambda$-calculus and the actual implementations of lazy functional languages.
While call-by-name re-evaluates an argument each time it is used --an operation
which is quite expensive-- call-by-need can be seen as a \emph{memoized}
version of call-by-name, where the value of an argument is stored the first
time it is evaluated for subsequent uses.  For example, if $t =
\appterm{\Delta}{(\appterm{I}{I})}$, where $\Delta =
\absterm{x}{\appterm{x}{x}}$ and $I = \absterm{z}{z}$, then call-by-name
duplicates the argument $\appterm{I}{I}$, while lazy languages first reduce
$\appterm{I}{I}$ to the value $I$ so that further uses of this argument do not
need to evaluate it again.

While the notion of needed reduction is defined with respect to (full strong)
\emph{normal forms}, call-by-need calculi evaluate programs to special values
called \emph{weak-head normal forms}, which are either abstractions or
arbitrary applications headed by a variable (\ie terms of the form
$\appterm{\appterm{x}{t_1}\ldots}{t_n}$ where $t_1 \ldots t_n$ are arbitrary
terms). To overcome this shortfall, we first adapt the notion of needed redex
to terms that are not going to be fully reduced to \emph{normal forms} but only
to \emph{weak-head normal forms}. Thus, informally, a redex $\pos{r}$ is
\emph{weak-head needed} in a term $t$ if $\pos{r}$ has to be contracted sooner
or later when reducing $t$ to a weak-head normal form. The derived notion of
strategy is called a \emph{weak-head needed strategy}, which always contracts a
weak-head needed redex.

This paper introduces two independent results about weak-head neededness, both
obtained by means of (non-idempotent) intersection
types~\cite{Gardner94,Carvalho:thesis} (a survey can be found
in~\cite{BucciarelliKV17}). We consider, in particular, typing system
$\V$~\cite{Kesner16} and show that it allows to identify all the weak-head
needed redexes of a weak-head normalising term. This is done by
adapting the classical notion of \emph{principal type}~\cite{Rocca88} and
proving that a redex in a weak-head normalising term $t$ is weak-head needed
iff it is typed in a principally typed derivation for $t$ in $\V$.

Our second goal is to show observational equivalence between call-by-need and
weak-head needed reduction. Two terms are observationally equivalent when all
the empirically testable computations on them are identical. This means that a
term $t$ can be evaluated to a weak-head normal form using the call-by-need
machinery if and only if the weak-head needed reduction normalises $t$.

By means of system $\V$ mentioned so far we use a technique to reason about
observational equivalence that is flexible, general and easy to verify or even
certify. Indeed, system $\V$ provides a semantic argument: first showing that a
term $t$ is typable in system $\V$ iff it is normalising for the weak-head
needed strategy ($t \in \WN{\weaknd}$), then by resorting to some results
in~\cite{Kesner16}, showing that system $\V$ is complete for call-by-name, \ie
a term $t$ is typable in system $\V$ iff $t$ is normalising for call-by-name
($t \in \WN{\callbyname}$); and that $t$ is normalising for call-by-name iff
$t$ is normalising for call-by-need ($t \in \WN{\callbyneed}$). Thus completing
the following chain of equivalences:
\vspace{-0.5em}
\begin{center}
\begin{tikzpicture}[xscale=1.5, yscale=1.1, baseline=(call-by-need)]
  \node(call-by-need) at(0,0) {$t \in \WN{\weaknd}$};
  \node(name) at(2,0){$t$ typable in $\V$};
  \node(types) at(4,0){$t \in \WN{\callbyname}$};
  \node(weakheadneeded) at(6,0){$t  \in \WN{\callbyneed}$};
  \draw[implies-implies,double equal sign distance] (call-by-need) to (name) node[pos=.5,right]{\scriptsize $ $};
  \draw[implies-implies,double equal sign distance] (name) to (types) node[pos=.5,right]{\scriptsize $ $};
  \draw[implies-implies,double equal sign distance] (types) to (weakheadneeded) node[pos=.5,right]{\scriptsize $ $};
\end{tikzpicture}
\label{fig:main-picture}
\end{center}
\vspace{-0.5em}
This leads to the observational equivalence between call-by-need, call-by-name
and weak-head needed reduction.

\emph{Structure of the paper}:
Sec.~\ref{sec:preliminaries} introduces preliminary concepts while
Sec.~\ref{sec:needed-notions} defines different notions of needed reduction.
The type system $\V$ is studied in Sec.~\ref{sec:systemV}.
Sec~\ref{sec:allowable} extends $\beta$-reduction to derivation trees. We show
in Sec.~\ref{sec:redexes} how system $\V$ identifies weak-head needed redexes,
while Sec.~\ref{sec:charact} gives a characterisation of normalisation for the
weak-head needed reduction. Sec.~\ref{sec:call-by-need} is devoted to define
call-by-need. Finally, Sec.~\ref{sec:results} presents the observational
equivalence result.

\section{Preliminaries}
\label{sec:preliminaries}

This section introduces some standard definitions and notions concerning the
reduction strategies studied in this paper, that is, call-by-name, head and
weak-head reduction, and neededness, this later notion being based on the
\emph{theory of residuals}~\cite{Barendregt84}.

\subsection{The Call-By-Name Lambda-Calculus}
\label{sec:callbyname}

Given a countable infinite set $\TermVariable$ of variables $x, y, z, \ldots$
we consider the following grammar:
\begin{center}
\begin{tabular}{rrcll}
\textbf{(Terms)}         & $t,u$      & $\Coloneq$ & $x \in \TermVariable \mid \appterm{t}{u} \mid \absterm{x}{t}$ \\
\textbf{(Values)}        & $v$        & $\Coloneq$ & $\absterm{x}{t}$ \\
\textbf{(Contexts)}      & $\ctxt{C}$ & $\Coloneq$ & $\Box \mid \appterm{\ctxt{C}}{t} \mid \appterm{t}{\ctxt{C}} \mid \absterm{x}{\ctxt{C}}$ \\
\textbf{(Name contexts)} & $\ctxt{E}$ & $\Coloneq$ & $\Box \mid \appterm{\ctxt{E}}{t}$
\end{tabular}
\end{center}

The set of $\lambda$-terms is denoted by $\Term$. We use $I$, $K$ and $\Omega$
to denote the terms $\absterm{x}{x}$, $\absterm{x}{\absterm{y}{x}}$ and
$\appterm{(\absterm{x}{\appterm{x}{x}})}{(\absterm{x}{\appterm{x}{x}})}$
respectively. We use $\ctxtapp{\ctxt{C}}{t}$ (resp. $\ctxtapp{\ctxt{E}}{t}$)
for the term obtained by replacing the hole $\Box$ of $\ctxt{C}$ (resp.
$\ctxt{E}$) by $t$. The sets of \emphdef{free} and \emphdef{bound variables} of
a term $t$, written respectively $\fv{t}$ and $\bv{t}$, are defined as
usual~\cite{Barendregt84}. We work with the standard notion of
\emphdef{$\alpha$-conversion}, \ie renaming of bound variables for
abstractions; thus for example $\absterm{x}{\appterm{x}{y}} \eqalpha
\absterm{z}{\appterm{z}{y}}$.

A term of the form $\appterm{(\absterm{x}{t})}{u}$ is called a
\emphdef{$\beta$-redex} (or just \emphdef{redex} when $\beta$ is clear from
the context) and $\lambda x$ is called the \emphdef{anchor} of the redex.
The \emphdef{one-step reduction relation} $\rewrite{\beta}$ (resp.
$\rewrite{\callbyname}$) is given by the closure by contexts $\ctxt{C}$
(resp. $\ctxt{E}$) of the rewriting rule $\appterm{(\absterm{x}{t})}{u} \rrule{\beta}
\substitute{x}{u}{t}$, where $\substitute{\_}{\_}{\_}$ denotes
the capture-free standard higher-order substitution. Thus, call-by-name
forbids reduction inside arguments and $\lambda$-abstractions, \eg
$\appterm{(\absterm{x}{II})}{(II)} \rewrite{\beta}
\appterm{(\absterm{x}{II})}{I}$ and $\appterm{(\absterm{x}{II})}{(II)} \rewrite{\beta}
\appterm{(\absterm{x}{I})}{(II)}$
but neither $\appterm{(\absterm{x}{II})}{(II)} \rewrite{\callbyname}
\appterm{(\absterm{x}{II})}{I}$ nor $\appterm{(\absterm{x}{II})}{(II)} \rewrite{\callbyname}
\appterm{(\absterm{x}{I})}{(II)}$ holds. We write $\rewriten{\beta}$
(resp. $\rewriten{\callbyname}$) for the reflexive-transitive closure of
$\rewrite{\beta}$ (resp. $\rewrite{\callbyname}$).

\subsection{Head, Weak-Head and Leftmost Reductions}
\label{sec:occurrences}

In order to introduce different notions of reduction, we start by formalising
the general mechanism of reduction which consists in contracting a redex at
some specific occurrence. \emphdef{Occurrences} are finite words over the
alphabet $\set{\pos{0}, \pos{1}}$. We use $\rootpos$ to denote the empty word
and notation $\pos{a^n}$ for $\pos{n} \in \mathbb{N}$ concatenations of some
letter $\pos{a}$ of the alphabet. The set of \emphdef{occurrences} of a given
term is defined by induction as follows: $\oc{x} \eqdef \set{\rootpos}$;
$\oc{\appterm{t}{u}} \eqdef \set{\rootpos} \cup \set{\pos{0p} \mid \pos{p} \in
\oc{t}} \cup \set{\pos{1p} \mid \pos{p} \in \oc{u}}$; $\oc{\absterm{x}{t}}
\eqdef \set{\rootpos} \cup \set{\pos{0p} \mid \pos{p} \in \oc{t}}$.

Given two occurrences $\pos{p}$ and $\pos{q}$, we use the notation $\pos{p}
\prefix \pos{q}$ to mean that $\pos{p}$ is a \emphdef{prefix} of $\pos{q}$,
\ie there is $\pos{p'}$ such that $\pos{p} \pos{p'} = \pos{q}$. We denote by
$\termat{t}{\pos{p}}$ the \emphdef{subterm of $t$ at occurrence $\pos{p}$},
defined as expected~\cite{BaaderN98}, thus for example
$\termat{(\appterm{(\absterm{x}{y})}{z})}{\pos{00}} = y$. The set of
\emphdef{redex occurrences} of $t$ is defined by $\roc{t} \eqdef \set{\pos{p}
\in \oc{t} \mid \termat{t}{\pos{p}} = \appterm{(\absterm{x}{s})}{u}}$. We use
the notation $\reduction{\pos{r}}{t}{t'}{\beta}$ to mean that $\pos{r} \in
\roc{t}$ and $t$ reduces to $t'$ by \emphdef{contracting} the redex at
occurrence $\pos{r}$, \eg
$\reduction{\pos{000}}{\appterm{(\absterm{x}{\appterm{\appterm{(\absterm{y}{y})}{x}}{x}})}{z}}{\appterm{(\absterm{x}{\appterm{x}{x}})}{z}}{\beta}$.
This notion is extended to reduction sequences as expected, and noted
$\reductionn{\rho}{t}{t'}{\beta}$, where $\rho$ is the list of all the redex
occurrences contracted along the reduction sequence. We use $\emptyred$ to
denote the empty reduction sequence, so that
$\reductionn{\emptyred}{t}{t}{\beta}$ holds for every term $t$.

Any term $t$ has exactly one of the following forms:
$\absterm{x_1}{\dots\absterm{x_n}{\appterm{y}{\appterm{t_1}{\appterm{\ldots}{t_m}}}}}$
or
$\absterm{x_1}{\dots\absterm{x_n}{\appterm{\appterm{(\absterm{y}{s})}{u}}{\appterm{t_1}{\appterm{\ldots}{t_m}}}}}$
with $n, m \geq 0$. In the latter case we say that
$\appterm{(\absterm{y}{s})}{u}$ is the \emphdef{head redex} of $t$, while in
the former case there is no head redex. Moreover, if $n = 0$, we say that
$\appterm{(\absterm{y}{s})}{u}$ is the \emphdef{weak-head redex} of $t$. In
terms of occurrences, the \emph{head redex} of $t$ is the \emph{minimal}
redex occurrence of the form $\pos{0^n}$ with $\pos{n} \geq 0$. In particular,
if it satisfies that $\termat{t}{\pos{0^k}}$ is not an abstraction for every
$\pos{k} \leq \pos{n}$, it is the \emph{weak-head redex} of $t$. A reduction
sequence contracting at each step the head redex (resp. weak-head redex) of the
corresponding term is called the \emphdef{head reduction} (resp.
\emphdef{weak-head reduction}).

Given two redex occurrences $\pos{r}, \pos{r'} \in \roc{t}$, we say that
$\pos{r}$ is \emphdef{to-the-left of} $ \pos{r'}$ if the anchor of $\pos{r}$ is
to the left of the anchor of $\pos{r'}$. Thus for example, the redex occurrence
$\pos{0}$ is to-the-left of $\pos{1}$ in the term
$\appterm{(\appterm{I}{x})}{(\appterm{I}{y})}$, and $\pos{\rootpos}$ is
to-the-left of $\pos{00}$ in $\appterm{(\absterm{x}{(\appterm{I}{I})})}{z}$.
Alternatively, the relation \emph{to-the-left} can be understood as a
dictionary order between redex occurrences, \ie $\pos{r}$ is
\emph{to-the-left of} $\pos{r'}$ if either $\pos{r' = rq}$ with $\pos{q} \neq
\rootpos$ (\ie $\pos{r}$ is a proper prefix of $\pos{r'}$); or $\pos{r} =
\pos{p0q}$ and $\pos{r'} = \pos{p1q'}$ (\ie they share a common prefix and
$\pos{r}$ is on the left-hand side of an application while $\pos{r'}$ is on the
right-hand side).  Notice that in any case this implies $\pos{r'} \not\prefix
\pos{r}$. Since this notion defines a total order on redexes, every term not in
normal form has a unique \emphdef{leftmost redex}. The term $t$
\emphdef{leftmost reduces} to $t'$ if $t$ reduces to $t'$ and the reduction
step contracts the leftmost redex of $t$. For example,
$\appterm{(\appterm{I}{x})}{(\appterm{I}{y})}$ leftmost reduces to
$\appterm{x}{(\appterm{I}{y})}$ and
$\appterm{(\absterm{x}{(\appterm{I}{I})})}{z}$ leftmost reduces to
$\appterm{I}{I}$. This notion extends to reduction sequences as expected.

\section{Towards neededness}
\label{sec:needed-notions}

Needed reduction is based on two fundamental notions: that of residual, which
describes how a given redex is traced all along a reduction sequence, and that
of normal form, which gives the form of the expected result of the reduction
sequence. This section extends the standard notion of needed
reduction~\cite{BarendregtKKS87} to those of head and weak-head needed
reductions.

\subsection{Residuals}
\label{sec:residuals}

Given a term $t$, $\pos{p} \in \oc{t}$ and $\pos{r} \in \roc{t}$, the
\emphdef{descendants of $\pos{p}$ after $\pos{r}$ in $t$}, written
$\residuals{\pos{p}}{\pos{r}}$, is the set of \emph{occurrences} defined as
follows: $$
\begin{array}{rl}
\varnothing
  & \text{if $\pos{p} = \pos{r}$ or $\pos{p} = \pos{r0}$} \\
\set{\pos{p}}
  & \text{if $\pos{r} \not\prefix \pos{p}$} \\
\set{\pos{rq}}
  & \text{if $\pos{p} = \pos{r00q}$} \\
\set{\pos{rkq} \mid \termat{s}{\pos{k}} = x}
  & \text{if $\pos{p} = \pos{r1q}$ with $\termat{t}{\pos{r}} = \appterm{(\absterm{x}{s})}{u}$}
\end{array} $$
For instance, given $t =
\appterm{(\absterm{x}{\appterm{(\absterm{y}{x})}{x}})}{z}$, then $\oc{t} =
\set{\rootpos, \pos{0}, \pos{1}, \pos{00}, \pos{000}, \pos{001}, \pos{0000}}$,
$\roc{t} = \set{\rootpos, \pos{00}}$, $\residuals{\pos{00}}{\pos{00}} =
\varnothing$, $\residuals{\rootpos}{\pos{00}} = \set{\rootpos}$,
$\residuals{\pos{00}}{\rootpos} = \set{\rootpos}$ and
$\residuals{\pos{1}}{\rootpos} = \set{\pos{1}, \pos{00}}$.

Notice that $\residuals{\pos{p}}{\pos{r}} \subseteq \oc{t'}$ where
$\reduction{\pos{r}}{t}{t'}{\beta}$. Furthermore, if $\pos{p}$ is the
occurrence of a redex in $t$ (\ie $\pos{p} \in \roc{t}$), then
$\residuals{\pos{p}}{\pos{r}} \subseteq \roc{t'}$, and each position in
$\residuals{\pos{p}}{\pos{r}}$ is called a \emphdef{residual} of $\pos{p}$
after reducing $\pos{r}$. This notion is extended to sets of redex
occurrences, indeed, the \emphdef{residuals of $\posset{P}$ after $\pos{r}$
in $t$} are $\residuals{\posset{P}}{\pos{r}} \eqdef \bigcup_{\pos{p} \in
\posset{P}}{\residuals{\pos{p}}{\pos{r}}}$. In particular
$\residuals{\varnothing}{\pos{r}} = \varnothing$. Given
$\reductionn{\rho}{t}{t'}{\beta}$ and $\posset{P} \subseteq \roc{t}$,
the \emphdef{residuals of $\posset{P}$ after the sequence $\rho$} are:
$\residuals{\posset{P}}{\emptyred} \eqdef \posset{P}$ and
$\residuals{\posset{P}}{\pos{r}\rho'} \eqdef
\residuals{(\residuals{\posset{P}}{\pos{r}})}{\rho'}$.

Stability of the to-the-left relation makes use of the notion of residual:

\begin{lemma}
\label{lem:leftPreservation}
Given a term $t$, let $\pos{l}, \pos{r}, \pos{s} \in \roc{t}$ such that
$\pos{l}$ is to-the-left of $\pos{r}$, $\pos{s} \not\prefix \pos{l}$ and
$\reduction{\pos{s}}{t}{t'}{\beta}$. Then, $\pos{l} \in \roc{t'}$ and
$\pos{l}$ is to-the-left of $\pos{r'}$ for every $\pos{r'} \in
\residuals{\pos{r}}{\pos{s}}$.
\end{lemma}

\begin{proof}
First, notice that $\pos{s} \not\prefix \pos{l}$ implies
$\residuals{\pos{l}}{\pos{s}} = \set{\pos{l}}$, then it is immediate to see
that $\pos{l} \in \roc{t'}$. If $\pos{s} \not\prefix \pos{r}$, then
$\residuals{\pos{r}}{\pos{s}} = \set{\pos{r}}$ and the result holds
immediately. Otherwise, $\pos{s} \prefix \pos{r}$ implies $\pos{s} \prefix
\pos{r'}$ for every $\pos{r'} \in \residuals{\pos{r}}{\pos{s}}$ by definition.
From $\pos{l}$ is to-the-left of $\pos{r}$, we may distinguish two cases:
either $\pos{r = lq}$ with $\pos{q} \neq \rootpos$; or $\pos{r} = \pos{p1q}$
and $\pos{l} = \pos{p0q'}$. Both cases imply that $\pos{l}$ and $\pos{r}$ share
a common prefix, say $\pos{p'}$ (resp. $\pos{l}$ or $\pos{p}$ on each case).
From $\pos{s} \not\prefix \pos{l}$ and $\pos{s} \prefix \pos{r}$ we know that
$\pos{p'}$ is a proper prefix of $\pos{s}$. Thus, it is a proper prefix of
$\pos{r'}$ too, \ie either $\pos{r'} = \pos{lq''}$ with $\pos{q''} \neq
\rootpos$; or $\pos{r'} = \pos{p1q''}$. Hence, $\pos{l}$ is to-the-left of
$\pos{r'}$ for every $\pos{r'} \in \residuals{\pos{r}}{\pos{s}}$.
\end{proof}

Notice that this result does not only implies that the leftmost redex is
preserved by reduction of other redexes, but also that the residual of the
leftmost redex occurs in exactly the same occurrence as the original one.

\begin{corollary}
\label{cor:leftmost}
Given a term $t$, and $\pos{l} \in \roc{t}$ the leftmost redex of $t$, if the
reduction $\reductionn{\rho}{t}{t'}{\beta}$ contracts neither $\pos{l}$ nor any
of its residuals, then $\pos{l} \in \roc{t'}$ is the leftmost redex of $t'$.
\end{corollary}

\begin{proof}
By induction on the length of $\rho$ using Lem.~\ref{lem:leftPreservation}.
\end{proof}



\subsection{Notions of Normal Form}
\label{sec:nf}

The expected result of evaluating a program is specified by means of some
appropriate notion of normal form. Given any relation $\rewrite{\R}$, a term
$t$ is said to be in \emphdef{$\R$-normal form} ($\NF{\R}$) iff there is no
$t'$ such that $t \rewrite{\R} t'$. A term $t$ is \emphdef{$\R$-normalising}
($\WN{\R}$) iff there exists $u \in \NF{\R}$ such that $t \rewriten{\R} u$.
Thus, given an $\R$-normalising term $t$, we can define the set of $\R$-normal
forms of $t$ as $\nf{\R}{t} \eqdef \set{t' \mid t \rewriten{\R} t' \land t' \in
\NF{\R}}$.

In particular, it turns out that a term in \emphdef{weak-head $\beta$-normal
form} ($\WHNF{\beta}$) is of the form
$\appterm{\appterm{\appterm{x}{t_1}}{\ldots}}{t_n}$ ($n \geq 0$) or
$\absterm{x}{t}$, where $t, t_1, \ldots, t_n$ are arbitrary terms, \ie it has
no weak-head redex. The set of weak-head $\beta$-normal forms of $t$ is
$\whnf{\beta}{t} \eqdef \set{t' \mid t \rewriten{\beta} t' \land t' \in
\WHNF{\beta}}$.

Similarly, a term in \emphdef{head $\beta$-normal form} ($\HNF{\beta}$) turns
out to be of the form
$\absterm{x_1}{\ldots\absterm{x_n}{\appterm{\appterm{\appterm{x}{t_1}}{\ldots}}{t_m}}}$ ($n,m \geq 0$),
\ie it has no head redex. The set of head $\beta$-normal forms of $t$ is
given by $\hnf{\beta}{t} \eqdef \set{t' \mid t \rewriten{\beta} t' \land t' \in
\HNF{\beta}}$.

Last, any term in \emphdef{$\beta$-normal form} ($\NF{\beta}$) has the form
$\absterm{x_1}{\ldots\absterm{x_n}{\appterm{x}{t_1 \ldots t_m}}}$ ($n, m \geq
0$) where $t_1, \ldots, t_m$ are themselves in $\beta$-normal form. It is
well-known that the set $\nf{\beta}{t}$ is a singleton, so we may use it
either as a set or as its unique element.

It is worth noticing that $\NF{\beta} \subset \HNF{\beta} \subset
\WHNF{\beta}$. Indeed, the inclusions are strict, for instance
$\absterm{x}{\appterm{(\absterm{y}{y})}{z}}$ is in weak-head but not in head
$\beta$-normal form, while
$\appterm{x}{\appterm{(\appterm{(\absterm{y}{y})}{x})}{z}}$ is in head but not
in $\beta$-normal form.


\subsection{Notions of Needed Reduction}
\label{sec:needed}

The different notions of normal form considered in Sec.~\ref{sec:nf} suggest
different notions of needed reduction, besides the standard one in the
literature~\cite{BarendregtKKS87}. Indeed, consider $\pos{r} \in \roc{t}$. We
say that $\pos{r}$ is \emphdef{used} in a reduction sequence $\rho$ iff $\rho$
reduces $\pos{r}$ or some residual of $\pos{r}$. Then:
\begin{enumerate}
  \item $\pos{r}$ is \emphdef{needed} in $t$ if every reduction sequence
  from $t$ to $\beta$-normal form uses $\pos{r}$;
  \item $\pos{r}$ is \emphdef{head needed} in $t$ if every reduction sequence
from $t$ to  head $\beta$-normal form uses $\pos{r}$;
  \item $\pos{r}$ is \emphdef{weak-head needed} in $t$ if every reduction
  sequence of $t$ to  weak-head $\beta$-normal form uses $\pos{r}$.
\end{enumerate}

Notice in particular that $\nf{\beta}{t} = \varnothing$ (resp. $\hnf{\beta}{t}
= \varnothing$ or $\whnf{\beta}{t} = \varnothing$) implies every redex in $t$
is needed (resp. head needed or weak-head needed).

A \emphdef{one-step reduction $\rewrite{\beta}$} is \emphdef{needed}
(resp. \emphdef{head} or \emphdef{weak-head needed}), noted  $\rewrite{\needed}$
(resp. $\rewrite{\headnd}$ or $\rewrite{\weaknd}$), if the contracted redex is needed 
(resp. head or weak-head needed).
A \emphdef{reduction sequence $\rewriten{\beta}$} is \emphdef{needed}
(resp. \emphdef{head} or \emphdef{weak-head needed}), noted $\rewriten{\needed}$
(resp. $\rewriten{\headnd}$ or $\rewriten{\weaknd}$), if every reduction step in the
sequence is needed (resp. head or weak-head needed).

For instance, consider the reduction sequence: $$
\appterm{(\absterm{y}{\absterm{x}{\appterm{\appterm{I}{x}}{(\underline{\appterm{I}{I}}_{\pos{r_1}})}}})}{(\appterm{I}{I})}
\rewrite{\needed} \appterm{(\absterm{y}{\absterm{x}{\appterm{\underline{\appterm{I}{x}}_{\pos{r_2}}}{I}}})}{(\appterm{I}{I})}
\rewrite{\needed} \underline{\appterm{(\absterm{y}{\absterm{x}{\appterm{x}{I}}})}{(\appterm{I}{I})}}_{\pos{r_3}}
\hspace{-.5em}\rewrite{\needed} \absterm{x}{\appterm{x}{I}} $$
which is needed but not head needed, since redex $\pos{r_1}$ might not be contracted to
reach a head normal form: $$
\appterm{(\absterm{y}{\absterm{x}{\appterm{\underline{\appterm{I}{x}}_{\pos{r_2}}}{(\appterm{I}{I})}}})}{(\appterm{I}{I})}
\rewrite{\headnd} \underline{\appterm{(\absterm{y}{\absterm{x}{\appterm{x}{(\appterm{I}{I})}}})}{(\appterm{I}{I})}}_{\pos{r_3}}
\hspace{-.5em}\rewrite{\headnd} \absterm{x}{\appterm{x}{(\appterm{I}{I})}} $$
Moreover, this second reduction sequence is head needed but not weak-head needed since
only redex $\pos{r_3}$ is needed to get a weak-head normal form: $$
\underline{\appterm{(\absterm{y}{\absterm{x}{\appterm{\appterm{I}{x}}{(\appterm{I}{I})}}})}{(\appterm{I}{I})}}_{\pos{r_3}}
\hspace{-.5em}\rewrite{\weaknd} \absterm{x}{\appterm{\appterm{I}{x}}{(\appterm{I}{I})}} $$

Notice that the following equalities hold: $\NF{\needed} = \NF{\beta}$,
$\NF{\headnd} = \HNF{\beta}$ and $\NF{\weaknd} = \WHNF{\beta}$.


Leftmost redexes and reduction sequences are indeed needed:

\begin{lemma}
\label{lem:leftmostNeeded}
The leftmost redex in any term not in normal form (resp. head or weak-head
normal form) is needed (resp. head or weak-head needed).
\end{lemma}

\begin{proof}
Since the existing proof~\cite{BarendregtKKS87} does not extend to weak-head
normal forms, we give here an alternative argument which does not only
cover the standard case of needed reduction but also the new ones of head and
weak-head reductions.

Let $t$ be a term not in normal form (resp. head normal form or weak-head
normal form) and let us consider $\reductionn{\rho}{t}{t'}{\beta}$ such that
$t'$ is a normal form (resp. head normal form or weak-head normal form). Assume
towards a contradiction that the leftmost redex $\pos{l}$ of $t$ is not used in
$\rho$. By Cor.~\ref{cor:leftmost} the occurrence $\pos{l}$ is still the
leftmost redex of $t'$. This leads to a contradiction with $t'$ being a normal
form (resp. head normal form or weak-head normal form), in particular because
the leftmost redex of a term not in head normal form (resp. weak-head normal
form) is necessarily the head redex (resp. weak-head redex).
\end{proof}

\begin{theorem}
\label{thm:leftmost}
Let $\pos{r} \in \roc{t}$ and $\reductionn{\rho}{t}{t'}{\beta}$ be the leftmost
reduction (resp. head reduction or weak-head reduction) starting with $t$ such
that $t' = \nf{\beta}{t}$ (resp. $t' \in \hnf{\beta}{t}$ or $t' \in
\whnf{\beta}{t}$). Then, $\pos{r}$ is needed (resp. head or weak-head needed) in
$t$ iff $\pos{r}$ is used in $\rho$.
\end{theorem}

\begin{proof}
$\Rightarrow)$ Immediate by definition of needed (resp. head or weak-head
needed).

$\Leftarrow)$ Let $\rho = \rho'\pos{r'}\rho''$ with $\pos{r'} \in
\residuals{\pos{r}}{\rho'}$. By hypothesis $\pos{r'}$ is the leftmost redex of
its corresponding term. By Lem.~\ref{lem:leftmostNeeded}, $\pos{r'}$ is needed
(resp. head or weak-head needed). Notice that, given a redex $\pos{s}$ not
needed in $t$, it follows from the definition that no residual of $\pos{s}$ is
needed either. Therefore, $\pos{r'}$ needed (resp. head or weak-head needed)
implies $\pos{r}$ needed (resp. head or weak-head needed) as well.
\end{proof}

Notice that the weak-head reduction is a prefix of the head reduction,
which is in turn a prefix of the leftmost reduction to normal form.
As a consequence, it is immediate to see that every weak-head needed
redex is in particular head needed, and every head needed redex is
needed as well.  For example, consider:
$$\underline{\appterm{(\absterm{y}{\absterm{x}{\appterm{\overline{\appterm{I}{x}}^\pos{r_2}}{(\overline{\appterm{I}{I}}^\pos{r_3})}}})}{(\overline{\appterm{I}{I}}^\pos{r_4})}}_\pos{r_1}$$
where $\pos{r_3}$ is a needed redex but not head needed nor
weak-head needed. However, $\pos{r_2}$ is both needed and
head needed, while $\pos{r_1}$ is the only weak-head needed redex in
the term, and $\pos{r_4}$ is not needed at all.


\section{The Type System $\V$}
\label{sec:systemV}

In this section we recall the (non-idempotent) intersection type system
$\V$~\cite{Kesner16} --an extension of those
in~\cite{Gardner94,Carvalho:thesis}-- used here to characterise normalising
terms w.r.t. the weak-head strategy. More precisely, we show that $t$ is
typable in system $\V$ if and only if $t$ is normalising when only weak-head
needed redexes are contracted. This characterisation is used in
Sec.~\ref{sec:results} to conclude that the weak-head needed strategy is
observationally equivalent to the call-by-need calculus (to be introduced in
Sec.~\ref{sec:call-by-need}).

Given a constant type $\valuetype$ that denotes \emph{answers} and a countable
infinite set $\TypeVariable$ of base type variables $\alpha, \beta, \gamma,
\ldots$, we define the following sets of types:
\begin{center}
\begin{tabular}{rrcll}
\textbf{(Types)}          & $\tau,\sigma$ & $\Coloneq$ & $\valuetype \mid \alpha \in \TypeVariable \mid \functtype{\M}{\tau}$ \\
\textbf{(Multiset types)} & $\M, \N$      & $\Coloneq$ & $\intertype{\tau_i}{i \in I}$ & where $I$ is a finite set
\end{tabular}
\end{center}

The empty multiset is denoted by $\intertype{}{}$. We remark that types are
\emph{strict}~\cite{Bakel92}, \ie the right-hand sides of functional types are
never multisets. Thus, the general form of a type is
$\functtype{\M_1}{\functtype{\ldots}{\functtype{\M_n}{\tau}}}$ with $\tau$
being the constant type or a base type variable.

\emphdef{Typing contexts} (or just \emphdef{contexts}), written $\Gamma,
\Delta$, are functions from variables to multiset types, assigning the empty
multiset to all but a finite set of variables. The domain of $\Gamma$ is given
by $\dom{\Gamma} \eqdef \set{x \mid \Gamma(x) \neq \intertype{}{}}$. The
\emphdef{union of contexts}, written $\ctxtsum{\Gamma}{\Delta}{}$, is defined
by $(\ctxtsum{\Gamma}{\Delta}{})(x) \eqdef \Gamma(x) \sqcup \Delta(x)$, where
$\sqcup$ denotes multiset union. An example is $(x:\intertype{\sigma}{},
y:\intertype{\tau}{}) + (x:\intertype{\sigma}{}, z:\intertype{\tau}{}) =
(x:\intertype{\sigma, \sigma}{}, y:\intertype{\tau}{}, z:\intertype{\tau}{})$.
This notion is extended to several contexts as expected, so that $+_{i \in I}
\Gamma_i$ denotes a finite union of contexts (when $I = \varnothing$ the
notation is to be understood as the empty context). We write
$\ctxtres{\Gamma}{x}{}$ for the context  $(\ctxtres{\Gamma}{x}{})(x) =
\intertype{}{}$ and $(\ctxtres{\Gamma}{x}{})(y) = \Gamma(y)$ if $y \neq x$.

\emphdef{Type judgements} have the form $\sequ{\Gamma}{\assign{t}{\tau}}$,
where $\Gamma$ is a typing context, $t$ is a term and $\tau$ is a type. The
intersection type system $\V$ for the $\lambda$-calculus is given in
Fig.~\ref{fig:typingSchemesV}.

\vspace{-1em}
\begin{figure}[h] $$
\begin{array}{c}
\Rule{\vphantom{\Gamma}}
     {\sequ{\assign{x}{\intertype{\tau}{}}}{\assign{x}{\tau}}}
     {\ruleAxiom}
\qquad
\Rule{\sequ{\Gamma}{\assign{t}{\tau}}}
     {\sequ{\ctxtres{\Gamma}{x}{}}{\assign{\absterm{x}{t}}{\functtype{\Gamma(x)}{\tau}}}}
     {\ruleArrowI}
\\
\\
\Rule{\vphantom{\Gamma}}
     {\sequ{}{\assign{\absterm{x}{t}}{\valuetype}}}
     {\ruleValue}
\qquad
\Rule{\sequ{\Gamma}{\assign{t}{\functtype{\intertype{\sigma_i}{i \in I}}{\tau}}}
      \quad
      (\sequ{\Delta_i}{\assign{u}{\sigma_i}})_{i \in I}}{\sequ{\ctxtsum{\Gamma}{\Delta_i}{i \in I}
     }
     {\assign{\appterm{t}{u}}{\tau}}}
     {\ruleArrowE}
\end{array} $$
\caption{The non-idempotent intersection type system $\V$.}
\label{fig:typingSchemesV}
\end{figure}
\vspace{-1em}

The constant type $\valuetype$ in rule $\ruleValue$ is used to type values. The
axiom $\ruleAxiom$ is relevant (there is no weakening) and the rule
$\ruleArrowE$ is multiplicative. Note that the argument of an application is
typed $\#(I)$ times by the premises of rule $\ruleArrowE$. A particular case is
when $I = \varnothing$: the subterm $u$ occurring in the typed term
$\appterm{t}{u}$ turns out to be untyped.

A \emphdef{(type) derivation} is a tree obtained by applying the
(inductive) typing rules of system $\V$. The notation
$\derivable{}{\sequ{\Gamma}{\assign{t}{\tau}}}{\V}$ means there is a
derivation of the judgement $\sequ{\Gamma}{\assign{t}{\tau}}$ in system
$\V$. The term $t$ is typable in system $\V$, or $\V$-typable, iff $t$
is the \emphdef{subject} of some derivation, \ie iff
there are $\Gamma$ and $\tau$ such that
$\derivable{}{\sequ{\Gamma}{\assign{t}{\tau}}}{\V}$. We use the
capital Greek letters $\Phi, \Psi, \ldots$ to name type derivations,
by writing for example $\derivable{\Phi}{\sequ{\Gamma}{\assign{t}{\tau}}}{\V}$.
For short, we usually denote with $\Phi_{t}$ a derivation with subject $t$
for some type and context.
The \emphdef{size of the derivation} $\Phi$, denoted by $\size{\Phi}$,
is defined as the number of nodes of the corresponding derivation tree.
We write $\projrule{\Phi} \in
\set{\ruleAxiom, \ruleArrowI, \ruleArrowE}$ to access the last rule
applied in the derivation $\Phi$. Likewise, $\projprem{\Phi}$ is
the \emph{multiset} of proper maximal subderivations of $\Phi$. For instance,
given $$
\Rule{
  \Phi_{t} \quad \many{\Phi_{u}^{i}}{i \in I}
}{
  \sequ{\Gamma}{\assign{\appterm{t}{u}}{\tau}}
}{
  \hspace{-7.5em} \Phi = 
  \hspace{5.5em} \ruleArrowE
} $$
we have $\projrule{\Phi} = \ruleArrowE$ and $\projprem{\Phi} =
\multiset{\Phi_{t}} \sqcup \multiset{\Phi_{u}^{i} \mid i \in I}$. We also use
functions $\projctxt{\Phi}$, $\projsubj{\Phi}$ and $\projtype{\Phi}$ to access
the context, subject and type of the judgement in the root of the derivation
tree respectively. For short, we also use notation $\Phi(x)$ to denote the
type associated to the variable $x$ in the typing environment of the conclusion
of $\Phi$ (\ie $\Phi(x) \eqdef \projctxt{\Phi}(x)$).

\bigskip
Intersection type systems can usually be seen as models~\cite{CoppoD80}, \ie
typing is stable by convertibility: if $t$ is typable and $t =_{\beta} t'$,
then $t'$ is typable too. This property splits in two different statements
known as \emph{subject reduction}  and {\it subject expansion} respectively,  the first
one giving stability of typing by reduction, the second one by expansion.
In the particular case of \emph{non-idempotent types}, subject reduction
refines to {\it weighted subject-reduction}, stating that not only typability
is stable by reduction, but also that the size of type derivations is
decreasing. Moreover, this decrease is strict when reduction is performed on
special occurrences of redexes, called \emph{typed occurrences}. We now
introduce all these concepts.

Given a type derivation $\Phi$, the set $\toc{\Phi}$ of \emphdef{typed occurrences}
of $\Phi$, which is a subset of $\oc{\projsubj{\Phi}}$, is defined by induction on
the last rule of $\Phi$.
\begin{itemize}
  \item If $\projrule{\Phi} \in \set{\ruleAxiom, \ruleValue}$, then $\toc{\Phi} \eqdef \set{\rootpos}$.
  \item If $\projrule{\Phi} = \ruleArrowI$ with $\projsubj{\Phi} = \absterm{x}{t}$ and $\projprem{\Phi} = \multiset{\Phi_{t}}$, then $\toc{\Phi} \eqdef \set{\rootpos} \cup \set{\pos{0p} \mid \pos{p} \in \toc{\Phi_{t}}}$.
  \item If $\projrule{\Phi} = \ruleArrowE$ with $\projsubj{\Phi} = \appterm{t}{u}$ and $\projprem{\Phi} = \multiset{\Phi_{t}} \sqcup \multiset{\Phi_{u}^{i} \mid i \in I}$, then $\toc{\Phi} \eqdef \set{\rootpos} \cup \set{\pos{0p} \mid \pos{p} \in \toc{\Phi_{t}}} \cup (\bigcup_{i \in I}{\set{\pos{1p} \mid \pos{p} \in \toc{\Phi_{u}^{i}}}})$.
\end{itemize}

Remark that there are two kind of untyped occurrences, those inside untyped
arguments of applications, and those inside untyped bodies of abstractions.
For instance consider the following type derivations: $$
\qquad\quad
\prooftree
  \prooftree
    \Rule{}{
      \sequ{\assign{x}{\intertype{\valuetype}{}}}{\assign{x}{\valuetype}}
    }{\ruleAxiom}
  \justifies
    \sequ{\assign{x}{\intertype{\valuetype}{}}}{\assign{\absterm{y}{x}}{\functtype{\intertype{}{}}{\valuetype}}}
  \using
    \ruleArrowI
  \endprooftree
\justifies
  \sequ{}{\assign{K}{\functtype{\intertype{\valuetype}{}}{\functtype{\intertype{}{}}{\valuetype}}}}
\using
  \hspace{-13.9em} \Phi_{K} = 
  \hspace{11.2em} \ruleArrowI
\endprooftree
\qquad
\qquad
\qquad\quad
\prooftree
  \prooftree
    \Phi_{K}
    \quad
    \Rule{}{
      \sequ{}{\assign{I}{\valuetype}}
    }{\ruleValue}
  \justifies
    \sequ{}{\assign{\appterm{K}{I}}{\functtype{\intertype{}{}}{\valuetype}}}
  \using
    \hspace{-10.6em} \Phi_{KI} = 
    \hspace{7.3em} \ruleArrowE
  \endprooftree
\justifies
  \sequ{}{\assign{\appterm{\appterm{K}{I}}{\Omega}}{\valuetype}}
\using
  \hspace{-11.2em} \Phi_{KI\Omega} = 
  \hspace{7.3em} \ruleArrowE
\endprooftree $$
Then, 
$\toc{\Phi_{KI\Omega}} = \set{\rootpos, \pos{0}, \pos{00}, \pos{01}, \pos{000},
\pos{0000}} \subseteq \oc{KI\Omega}$.

\begin{remark}
\label{rem:typedLeftmost}
The weak-head redex of a typed term is always a typed occurrence.
\end{remark}

For convenience we introduce an alternative way to denote type derivations.
We refer to $\applyax{x}{\tau}$ as the result of applying $\ruleAxiom$ with
subject $x$ and type $\tau$: $$
\applyax{x}{\tau} \eqdef
\Rule{\vphantom{\Phi}}
     {\sequ{\assign{x}{\intertype{\tau}{}}}{\assign{x}{\tau}}}
     {\ruleAxiom} $$
We denote with $\applyval{x}{t}$ the result of applying $\ruleValue$
abstracting $x$ and term $t$: $$
\applyval{x}{t} \eqdef
\Rule{\vphantom{\Phi}}
     {\sequ{}{\assign{\absterm{x}{t}}{\valuetype}}}
     {\ruleValue} $$
We refer to $\applyabs{x}{\Phi_{t}}$ as the result of applying $\ruleArrowI$
with premise $\Phi_{t}$ and abstracting variable $x$: $$
\applyabs{x}{\Phi_{t}} \eqdef
\Rule{\Phi_{t}}
     {\sequ{\ctxtres{\projctxt{\Phi_{t}}}{x}{}}{\assign{\absterm{x}{t}}{\functtype{\Phi_{t}(x)}{\projtype{\Phi_{t}}}}}}
     {\ruleArrowI} $$
Likewise, we write $\applyapp{\Phi_{t}}{u}{\many{\Phi_{u}^{i}}{i \in I}}$ for
the result of applying $\ruleArrowE$ with premises $\Phi_{t}$ and
$\many{\Phi_{u}^{i}}{i \in I}$, and  argument $u$ (the argument $u$ is untyped
when $I = \varnothing$). Note that this application is valid provided that
$\projtype{\Phi_{t}} = \functtype{\intertype{\sigma_i}{i \in I}}{\tau}$ and
$\many{\projtype{\Phi_{u}^{i}} = \sigma_i}{i \in I}$. Then: $$
\applyapp{\Phi_{t}}{u}{\many{\Phi_{u}^{i}}{i \in I}} \eqdef
\Rule{\Phi_{t}
      \qquad
      \many{\Phi_{u}^{i}}{i \in I}}
     {\sequ{\ctxtsum{\projctxt{\Phi_{t}}}{\projctxt{\Phi_{u}^{i}}}{i \in I}}{\assign{\appterm{t}{u}}{\tau}}}
     {\ruleArrowE} $$

Given $\Phi$ and $\pos{p} \in \toc{\Phi}$, the \emph{multiset}
$\treeat{\Phi}{\pos{p}}$ of \emphdef{all the subderivations of $\Phi$ at
occurrence $\pos{p}$} is inductively defined as follows:
\begin{itemize}
  \item If $\pos{p} = \rootpos$, then $\treeat{\Phi}{\pos{p}} \eqdef
  \multiset{\Phi}$.
  
  \item If $\pos{p} = \pos{0p'}$, \ie $\projrule{\Phi} \in \set{\ruleArrowI,
  \ruleArrowE}$ with $\projsubj{\Phi} \in \set{\absterm{x}{t}, \appterm{t}{u}}$
  and $\Phi_{t} \in \projprem{\Phi}$. Then, $\treeat{\Phi}{\pos{p}} \eqdef
  \treeat{\Phi_{t}}{\pos{p'}}$.
  
  \item If $\pos{p} = \pos{1p'}$, \ie $\projrule{\Phi} = \ruleArrowE$ with
  $\projsubj{\Phi} = \appterm{t}{u}$ and $\projprem{\Phi} = \multiset{\Phi_{t}}
  \sqcup \multiset{\Phi_{u}^{i} \mid i \in I}$. Then, $\treeat{\Phi}{\pos{p}}
  \eqdef \bigsqcup_{i \in  I}{\treeat{\Phi_{u}^{i}}{\pos{p'}}}$ (recall
  $\sqcup$ denotes multiset union).
\end{itemize}

Given type derivations $\Phi, \many{\Psi_{i}}{i \in I}$ and a position $\pos{p}
\in \oc{\projsubj{\Phi}}$, \emphdef{replacing the subderivations of $\Phi$ at
occurrence $\pos{p}$ by $\many{\Psi_{i}}{i \in I}$}, written
$\replaceat{\Phi}{\pos{p}}{\many{\Psi_{i}}{i \in I}}$, is a type derivation
inductively defined as follows, assuming $\#(\treeat{\Phi}{\pos{p}}) = \#(I)$
and $\projsubj{\Psi_{i}} = \projsubj{\Psi_{j}}$ for every $i,j \in I$ (we call
$s$ the unique subject of all these derivations):
\begin{itemize}
  \item If $\pos{p} = \pos{\rootpos}$, then
  $\replaceat{\Phi}{\pos{p}}{\many{\Psi_{i}}{i \in I}} \eqdef \Psi_{i_0}$ with
  $I = \set{i_0}$.
  
  \item If $\pos{p} = \pos{0p'}$, either:
  \begin{itemize}
    \item $\projrule{\Phi} = \ruleArrowI$ with $\projsubj{\Phi} =
    \absterm{x}{t}$ and $\projprem{\Phi} = \multiset{\Phi_{t}}$. Then,
    $\replaceat{\Phi}{\pos{p}}{\many{\Psi_{i}}{i \in I}} \eqdef
    \applyabs{x}{\replaceat{\Phi_{t}}{\pos{p'}}{\many{\Psi_{i}}{i \in I}}}$; or
    
    \item $\projrule{\Phi} = \ruleArrowE$ with $\projsubj{\Phi} =
    \appterm{t}{u}$ and $\projprem{\Phi} = \multiset{\Phi_{t}} \sqcup
    \multiset{\Phi_{u}^{j} \mid j \in J}$. Then,
    $\replaceat{\Phi}{\pos{p}}{\many{\Psi_{i}}{i \in I}} \eqdef
    \applyapp{\replaceat{\Phi_{t}}{\pos{p'}}{\many{\Psi_{i}}{i \in
    I}}}{u}{\many{\Phi_{u}^{j}}{j \in J}}$.
  \end{itemize}
  
  \item If $\pos{p} = \pos{1p'}$, \ie $\projrule{\Phi} = \ruleArrowE$ with
  $\projsubj{\Phi} = \appterm{t}{u}$, $\projprem{\Phi} = \multiset{\Phi_{t}}
  \sqcup \multiset{\Phi_{u}^{j} \mid j \in J}$ and $I = \biguplus_{j \in
  J}{I_j}$. Then, $$\replaceat{\Phi}{\pos{p}}{\many{\Psi_{i}}{i \in I}} \eqdef
  \applyapp{\Phi_{t}}{\replaceat{u}{\pos{p'}}{s}}{\many{\replaceat{\Phi_{u}^{j}}{\pos{p'}}{\many{\Psi_{i}}{i \in I_j}}}{j \in J}}$$
  where 
$s$ is the unique subject of all the derivations $\Psi_{i}$ and
$\replaceat{r'}{\pos{q}}{r}$ denotes the replacement of the subterm
$\termat{r'}{\pos{q}}$ by $r$ in $r'$ (variable capture is allowed).
Remark that the  decomposition of $I$ into the sets $I_j$ ($j \in J$) is
non-deterministic, thus replacement turns out to be a non-deterministic
operation.
\end{itemize}


We can now state the two main properties of system $\V$, whose proofs
can be found in Sec. 7 of~\cite{BucciarelliKV17}.

\begin{theorem}[Weighted Subject Reduction]
\label{thm:subjectReductionV}
Let $\derivable{\Phi}{\sequ{\Gamma}{\assign{t}{\tau}}}{\V}$. If
$\reduction{\pos{r}}{t}{t'}{\beta}$, then there exists $\Phi'$ s.t.
$\derivable{\Phi'}{\sequ{\Gamma}{\assign{t'}{\tau}}}{\V}$. Moreover,
\begin{enumerate}
 \item If $\pos{r} \in \toc{\Phi}$, then $\size{\Phi} > \size{\Phi'}$.
 \item If $\pos{r} \notin \toc{\Phi}$, then $\size{\Phi} = \size{\Phi'}$.
\end{enumerate}
\end{theorem}


\begin{theorem}[Subject Expansion]
\label{thm:subjectExpansionV}
Let $\derivable{\Phi'}{\sequ{\Gamma}{\assign{t'}{\tau}}}{\V}$. If $t
\rewrite{\beta} t'$, then there exists $\Phi$ s.t.
$\derivable{\Phi}{\sequ{\Gamma}{\assign{t}{\tau}}}{\V}$.
\end{theorem}


Note that weighted subject reduction implies that 
reduction of typed redex occurrences turns out to be normalising.


\section{Substitution and Reduction on Derivations}
\label{sec:allowable}

In order to relate typed redex occurrences of convertible terms, we now extend
the notion of $\beta$-reduction to derivation trees, by making use of a natural
and basic concept of typed substitution. In contrast to substitution and
$\beta$-reduction on {\it terms}, these operations are now both
non-deterministic on derivation trees (see~\cite{Vial:thesis} for discussions
and examples). Given a variable $x$ and type derivations $\Phi_{t}$ and
$\many{\Phi_{u}^{i}}{i \in I}$, the \emphdef{typed substitution} of\ $x$ by
$\many{\Phi_{u}^{i}}{i \in I}$ in $\Phi_{t}$, written
$\substitute{x}{\many{\Phi_{u}^{i}}{i \in I}}{\Phi_{t}}$ by making an abuse of
notation, is a type derivation inductively defined on $\Phi_{t}$, only if
$\Phi_{t}(x) = \intertype{\projtype{\Phi_{u}^{i}}}{i \in I}$:
\begin{itemize}
  \item If $\Phi_{t} = \applyax{y}{\tau}$, then $$
\begin{array}{rcl}
\substitute{x}{\many{\Phi_{u}^{i}}{i \in I}}{\Phi_{t}} & \eqdef &
\begin{cases}
\Phi_{u}^{i_0} & \text{if $y = x$ where $I = \set{i_0}$} \\
\Phi_{t}       & \text{if $y \neq x$}
\end{cases}
\end{array} $$
  
  \item If $\Phi_{t} = \applyval{y}{t'}$, then $$
\begin{array}{rcl}
\substitute{x}{\many{\Phi_{u}^{i}}{i \in I}}{\Phi_{t}} & \eqdef &
\applyval{y}{\substitute{x}{u}{t'}}
\end{array} $$
  if there is no capture of the variable $y$.
  
  \item If $\Phi_{t} = \applyabs{y}{\Phi_{t'}}$, then $$
\begin{array}{rcl}
\substitute{x}{\many{\Phi_{u}^{i}}{i \in I}}{\Phi_{t}} & \eqdef & \applyabs{y}{\substitute{x}{\many{\Phi_{u}^{i}}{i \in I}}{\Phi_{t'}}}
\end{array} $$
  if there is no capture of the variable $y$.
  
  \item If $\Phi_{t} = \applyapp{\Phi_{r}}{s}{\many{\Phi_{s}^{j}}{j \in J}}$ with $I = I' \uplus (\biguplus_{j \in J}{I_j})$, where
  $\Phi_{r}(x) = \intertype{\projtype{\Phi_{u}^{i}}}{i \in I'}$ and
  $\many{\Phi_{s}(x) = \intertype{\projtype{\Phi_{u}^{i}}}{i \in I_j}}{j
  \in J}$, then $$
\begin{array}{rcl}
\substitute{x}{\many{\Phi_{u}^{i}}{i \in I}}{\Phi_{t}} & \eqdef &
\applyapp{\substitute{x}{\many{\Phi_{u}^{i}}{i \in I'}}{\Phi_{r}}}{\substitute{x}{u}{s}}{\many{\substitute{x}{\many{\Phi_{u}^{i}}{i \in I_j}}{\Phi_{s}^{j}}}{j \in J}}
\end{array} $$
Remark that the decomposition of $I$ into $I'$ and the sets $I_j$ ($j \in J$)
is non-deterministic, thus substitution of derivation trees turns out to be a
non-deterministic operation.
\end{itemize}

Intuitively, the typed substitution replaces typed occurrences of $x$ in
$\Phi_t$ by a corresponding derivation $\Phi_u^i$ matching the same type, where
such a matching is chosen in a non-deterministic way. Moreover, it also
substitutes all untyped occurrences of $x$ by $u$, where this untyped operation
is completely deterministic. Thus, for example, consider the following
substitution, where $\Phi_{KI}$ is defined in Sec.~\ref{sec:systemV}: $$
\substitute{x}{\Phi_{KI}}{\left(
  \prooftree \Rule{}{
    \sequ{\assign{x}{\intertype{\functtype{\intertype{}{}}{\valuetype}}{}}}{\assign{x}{\functtype{\intertype{}{}}{\valuetype}}}
  }{\ruleAxiom} \justifies
  \sequ{\assign{x}{\intertype{\functtype{\intertype{}{}}{\valuetype}}{}}}{\assign{\appterm{x}{x}}{\valuetype}}
  \using \ruleArrowE \endprooftree \right)} = \prooftree \Phi_{KI}
\justifies \sequ{}{\assign{\appterm{(KI)}{(KI)}}{\valuetype}} \using
\ruleArrowE \endprooftree $$

%

The following lemma relates the typed occurrences of the trees composing a
substitution and those of the substituted tree itself:
  
\begin{lemma}
\label{lem:tocSubstitution}
Let $\Phi_{t}$ and $\many{\Phi_{u}^{i}}{i \in I}$ be derivations such that $\substitute{x}{\many{\Phi_{u}^{i}}{i \in I}}{\Phi_{t}}$ is defined, and $\pos{p} \in \oc{t}$. Then,
\begin{enumerate}
  \item\label{lem:tocSubstitution:i} $\pos{p} \in \toc{\Phi_{t}}$ iff $\pos{p} \in \toc{\substitute{x}{\many{\Phi_{u}^{i}}{i \in I}}{\Phi_{t}}}$.
  \item\label{lem:tocSubstitution:ii} $\pos{q} \in \toc{\Phi_{u}^{k}}$ for some $k \in I$ iff there exists $\pos{p} \in \toc{\Phi_{t}}$ such that $\termat{t}{\pos{p}} = x$ and $\pos{pq} \in \toc{\substitute{x}{\many{\Phi_{u}^{i}}{i \in I}}{\Phi_{t}}}$.
\end{enumerate}
\end{lemma}

\begin{proof}
By induction on $\Phi_{t}$.
\end{proof}

Based on the previous notion of substitutions on derivations, we are now able
to introduce (non-deterministic) reduction on derivation trees. The
\emphdef{reduction relation} $\rewrite{\beta}$ on derivation trees is then
defined by first considering the following basic rewriting rules.
\begin{enumerate}
  \item For typed $\beta$-redexes: $$
\prooftree
  \Rule{
    \derivable{\Phi_{t}}{\sequ{\Gamma;\assign{x}{\intertype{\sigma_i}{i \in I}}}{\assign{t}{\tau}}}{\V}
  }{
    \sequ{\Gamma}{\assign{\absterm{x}{t}}{\functtype{\intertype{\sigma_i}{i \in I}}{\tau}}}
  }{}
  \quad
  \many{\derivable{\Phi_{u}^{i}}{\sequ{\Delta_i}{u : \sigma_i}}{\V}}{i \in I}
\justifies
  \sequ{\ctxtsum{\Gamma}{\Delta_i}{i \in I}}{\assign{\appterm{(\absterm{x}{t})}{u}}{\tau}}
\using
  \ 
  \rrule{\beta}
  \ 
  \substitute{x}{\many{\Phi_{u}^{i}}{i \in I}}{\Phi_{t}}
\endprooftree $$

  \item For $\beta$-redexes in untyped occurrences, with $u \rewrite{\beta}
  u'$: $$
\begin{array}{c@{\qquad\quad}c}
\Rule{
  \sequ{\Gamma}{\assign{t}{\functtype{\intertype{}{}}{\tau}}}
}{
  \sequ{\Gamma}{\assign{\appterm{t}{u}}{\tau}}
}{}
\ 
\rrule{\nu}
\ 
\Rule{
  \sequ{\Gamma}{\assign{t}{\functtype{\intertype{}{}}{\tau}}}
}{
  \sequ{\Gamma}{\assign{\appterm{t}{u'}}{\tau}}
}{} & 
\Rule{
  \vphantom{\Gamma}
}{
  \sequ{}{\assign{\absterm{x}{u}}{\valuetype}}
}{}
\ 
\rrule{\xi}
\ 
\Rule{
  \vphantom{\Gamma}
}{
  \sequ{}{\assign{\absterm{x}{u'}}{\valuetype}}
}{}
\end{array} $$
\end{enumerate}

As in the case of the $\lambda$-calculus, where reduction is closed under usual
\emph{term} contexts, we need to close the previous relation under
\emph{derivation tree} contexts. However, a one-step reduction on a given
subterm causes many one-step reductions in the corresponding derivation tree
(recall $\treeat{\Phi}{\pos{p}}$ is defined to be a multiset). Then,
informally, given a redex occurrence $\pos{r}$ of $t$, a type derivation $\Phi$
of $t$, and the multiset of minimal subderivations of $\Phi$ containing
$\pos{r}$, written $\mathscr{M}$, we apply the reduction rules
$\rrule{\beta,\nu,\xi}$ to all the elements of $\mathscr{M}$, thus obtaining a
multiset $\mathscr{M'}$, and we recompose the type derivation of the reduct of
$t$.

To formalise this idea, given a type derivation $\Phi$ and an occurrence
$\pos{p} \in \oc{\projsubj{\Phi}}$, we define the \emphdef{maximal typed
prefix} of $\pos{p}$ in $\Phi$, written $\mtp{\pos{p}}{\Phi}$, as the unique
prefix of $\pos{p}$ satisfying
$$\mtp{\pos{p}}{\Phi} \in \toc{\Phi} \land \forall \pos{q} \in \toc{\Phi}.
(\pos{q} \prefix \pos{p} \implies \pos{q} \prefix \mtp{\pos{p}}{\Phi})$$
Notice that the multiset of subderivations of $\Phi$ at position
$\mtp{\pos{p}}{\Phi}$ (\ie $\treeat{\Phi}{\mtp{\pos{p}}{\Phi}}$) corresponds to
the multiset of minimal subderivations of $\Phi$ containing $\pos{p}$.
For instance, consider the type derivation $\Phi_{KI\Omega}$ presented in
Sec.~\ref{sec:systemV}: $$
\prooftree
  \prooftree
    \Phi_{K}
    \qquad
		\qquad
    \Rule{}{
      \sequ{}{\assign{I}{\valuetype}}
    }{
      \hspace{-5.8em} \Phi_{I} = 
      \hspace{3.3em} \ruleValue
	  }
  \justifies
    \sequ{}{\assign{\appterm{K}{I}}{\functtype{\intertype{}{}}{\valuetype}}}
  \using
    \ruleArrowE
  \endprooftree
\justifies
  \sequ{}{\assign{\appterm{\appterm{K}{I}}{\Omega}}{\valuetype}}
\using
  \hspace{-11.2em} \Phi_{KI\Omega} = 
  \hspace{7.3em} \ruleArrowE
\endprooftree $$
where $\toc{\Phi_{KI\Omega}} = \set{\rootpos, \pos{0}, \pos{00}, \pos{01},
\pos{000}, \pos{0000}}$. Indeed, $\mtp{\pos{010}}{\Phi_{KI\Omega}} = \pos{01}$
and the minimal subderivation of $\Phi_{KI\Omega}$ containing this occurrence
is $\treeat{\Phi_{KI\Omega}}{\pos{01}} = \multiset{\Phi_{I}}$. Also,
$\mtp{\pos{1p}}{\Phi_{KI\Omega}} = \rootpos$ where
$\treeat{\Phi_{KI\Omega}}{\rootpos} = \multiset{\Phi_{KI\Omega}}$.

Then, given terms $t$ and $t'$, and type derivations $\Phi$ and
$\Phi'$ of $t$ and $t'$ respectively, we say that $\Phi$ \emphdef{reduces} to
$\Phi'$ (written $\Phi \rewrite{\beta} \Phi'$) iff there exists $\pos{r} \in
\roc{t}$ such that
$$\treeat{\Phi}{\mtp{\pos{r}}{\Phi}} \leadsto_{\beta,\nu,\xi} \mathscr{M}
\quad\text{and}\quad
\Phi' = \replaceat{\Phi}{\mtp{\pos{r}}{\Phi}}{\mathscr{M}}$$
where $\leadsto_{\beta,\nu,\xi}$ denotes the lifting to multisets of the basic
rewriting rules introduced above by applying the same rule to all the elements
of the multiset.

This gives the reduction relation $\rewrite{\beta}$ on trees. A reduction
sequence on derivation trees contracting only redexes in typed positions is
dubbed a \emphdef{typed reduction sequence}.

Note that typed reductions are normalising by Thm.~\ref{thm:subjectReductionV},
yielding a special kind of derivation. Indeed, given a type derivation
$\derivable{\Phi}{\sequ{\Gamma}{\assign{t}{\tau}}}{\V}$, we say that $\Phi$ is
\emphdef{normal} iff $\toc{\Phi} \cap \roc{t} = \varnothing$. Reduction on
trees induces reduction on terms: when $\reductionn{\rho}{\Phi}{\Phi'}{\beta}$,
then $\projsubj{\Phi} \rewriten{\beta} \projsubj{\Phi'}$. By abuse of notation
we may denote both sequences with the same letter $\rho$.


\section{Weak-Head Neededness and Typed Occurrences}
\label{sec:redexes}

This section presents one of our main results. It establishes a connection
between  weak-head needed redexes and typed redex occurrences. More precisely,
we first show in Sec.~\ref{s:whn-typed} that every weak-head needed redex
occurrence turns out to be  a typed occurrence, whatever its type derivation
is. The converse does not however hold. But, we show in
Sec.~\ref{s:principal-whn} that any typed occurrence in a special kind of typed
derivation (that we call principal) corresponds to a weak-head needed redex
occurrence. We start with a technical lemma.

\begin{lemma}
\label{lem:typedDescendant}
Let $\reduction{\pos{r}}{\Phi_{t}}{\Phi_{t'}}{\beta}$ and $\pos{p} \in \oc{t}$ such that $\pos{p} \neq \pos{r}$ and $\pos{p} \neq \pos{r0}$. Then, $\pos{p} \in \toc{\Phi_{t}}$ iff there exists $\pos{p'} \in
\residuals{\pos{p}}{\pos{r}}$ such that $\pos{p'} \in \toc{\Phi_{t'}}$.
\end{lemma}

\begin{proof}
If $\pos{p} = \rootpos$ the result holds since $\residuals{\rootpos}{\pos{r}} =
\set{\rootpos}$ and $\rootpos \in \toc{\Phi}$ for every possible $\Phi$ by
definition. Then, assume $\pos{p} \neq \rootpos$. We proceed by induction on
$\pos{r}$.
\begin{itemize}
  \item $\pos{r} = \rootpos$. Then, $\pos{r} \prefix \pos{p}$ and $\pos{r} \in
  \roc{t} \cap \toc{\Phi_{t}}$. Moreover, $t = \termat{t}{\pos{r}} =
  \appterm{(\absterm{x}{t_1})}{t_2}$ and $$
\begin{array}{c@{\hskip 0.4em}c@{\hskip 0.3em}c}
\raisebox{-0.7em}{$
\Phi_{t} =
$} &
\Rule{
  \Rule{
    \Phi_{t_1}
  }{
    \sequ{\Gamma'}{\assign{\absterm{x}{t_1}}{\functtype{\intertype{\sigma_i}{i \in I}}{\tau}}}
  }{}
  \quad
  \many{\Phi_{t_2}^{i}}{i \in I}
}{
  \sequ{\Gamma}{\assign{\appterm{(\absterm{x}{t_1})}{t_2}}{\tau}}
}{} &
\raisebox{-0.7em}{$
\rewrite{\beta} \substitute{x}{\many{\Phi_{t_2}^{i}}{i \in I}}{\Phi_{t_1}} = \Phi_{t'}
$}
\end{array} $$ with $t' = \substitute{x}{t_2}{t_1}$. Then, there are two
possibilities for $\pos{p}$:
  \begin{enumerate}
    \item $\pos{p} = \pos{00p'}$. Then, $\residuals{\pos{p}}{\rootpos} =
    \set{\pos{p'}}$. By definition, $\pos{p} \in \toc{\Phi_{t}}$ iff $\pos{p'}
    \in \toc{\Phi_{t_1}}$. Moreover, by Lem.~\ref{lem:tocSubstitution}
    (\ref{lem:tocSubstitution:i}), $\pos{p'} \in \toc{\Phi_{t_1}}$ iff
    $\pos{p'} \in \toc{\Phi_{t'}}$ (notice that $\pos{p'} \in \oc{t_1}$). Thus,
    we conclude.
    \item $\pos{p} = \pos{1p''}$. Then, $\residuals{\pos{p}}{\rootpos} =
    \set{\pos{qp''} \mid \termat{t_1}{\pos{q}} = x}$. By definition $\pos{p} \in
    \toc{\Phi_{t}}$ iff $\pos{p''} \in \toc{\Phi_{t_2}^{k}}$ for some $k \in I$.
    By Lem.~\ref{lem:tocSubstitution} (\ref{lem:tocSubstitution:ii}), $\pos{p''}
    \in \toc{\Phi_{t_2}^{k}}$ iff there exists $\pos{q} \in \toc{\Phi_{t_1}}$
    such that $\termat{t_1}{\pos{q}} = x$ and $\pos{p'} = \pos{qp''} \in
    \toc{\Phi_{t'}}$. Thus, we conclude.
  \end{enumerate}
  \item $\pos{r} = \pos{0r'}$. Then, we analyse the form of $t$:
  \begin{itemize}
    \item $t = \appterm{t_1}{t_2}$. Then, $t' = \appterm{t'_1}{t_2}$ with
    $\reduction{\pos{r'}}{t_1}{t'_1}{\beta}$. Moreover, $$\Phi_{t} =
    \Rule{\Phi_{t_1} \quad \many{\Phi_{t_2}^{i}}{i \in
    I}}{\sequ{\Gamma}{\assign{\appterm{t_1}{t_2}}{\tau}}}{} \rewrite{\beta}
    \Rule{\Phi_{t'_1} \quad \many{\Phi_{t_2}^{i}}{i \in
    I}}{\sequ{\Gamma}{\assign{\appterm{t'_1}{t_2}}{\tau}}}{} = \Phi_{t'}$$ with
    $\Phi_{t_1} \rewrite{\beta} \Phi_{t'_1}$.
    \\
    Now we have two possibilities for $\pos{p}$:
    \begin{itemize}
      \item $\pos{p} = \pos{0q}$. By definition, $\pos{p} \in \toc{\Phi_{t}}$
      iff $\pos{q} \in \toc{\Phi_{t_1}}$. Since $\pos{p} \neq \pos{r}$ implies
      $\pos{q} \neq \pos{r'}$, by inductive hypothesis, $\pos{q} \in
      \toc{\Phi_{t_1}}$ iff there exists $\pos{q'} \in
      \residuals{\pos{q}}{\pos{r'}}$ such that $\pos{q'} \in
      \toc{\Phi_{t'_1}}$. Thus, by definition once again, $\pos{p} \in
      \toc{\Phi_{t}}$ iff there exists $\pos{p'} \in
      \residuals{\pos{p}}{\pos{r}}$ such that $\pos{p'} \in \toc{\Phi_{t'}}$
      (\ie $\pos{p'} = \pos{0q'}$).
      \item $\pos{p} = \pos{1q}$. Then, $\pos{r} \not\prefix \pos{p}$ and
      $\residuals{\pos{p}}{\pos{r}} = \set{\pos{p}}$ (\ie $\pos{p'} = \pos{p}$).
      Thus, $\pos{p} \in \toc{\Phi_{t}}$ iff $\pos{q} \in \toc{\Phi_{t_2}^{k}}$
      for some $k \in I$ iff $\pos{p'} \in \toc{\Phi_{t'}}$ by definition.
    \end{itemize}
    \item $t = \absterm{x}{t_1}$. Then, $t' = \absterm{x}{t'_1}$ with
    $\reduction{\pos{r'}}{t_1}{t'_1}{\beta}$. If $\projrule{\Phi_{t}} =
    \ruleValue$, the result is immediate since $\toc{\Phi_{t}} = \set{\rootpos}
    = \toc{\Phi_{t'}}$. If not, $$\Phi_{t} =
    \Rule{\Phi_{t_1}}{\sequ{\Gamma}{\assign{\absterm{x}{t_1}}{\tau}}}{}
    \rewrite{\beta}
    \Rule{\Phi_{t'_1}}{\sequ{\Gamma}{\assign{\absterm{x}{t'_1}}{\tau}}}{}
    = \Phi_{t'}$$ with $\Phi_{t_1} \rewrite{\beta} \Phi_{t'_1}$. Then, $\pos{p}
    = \pos{0q}$ and, by definition, $\pos{p} \in \toc{\Phi_{t_1}}$ iff
    $\pos{q} \in \toc{\Phi_{t_1}}$. By inductive hypothesis, $\pos{q} \in
    \toc{\Phi_{t_1}}$ iff there exists $\pos{q'} \in
    \residuals{\pos{q}}{\pos{r'}}$ such that $\pos{q'} \in \toc{\Phi_{t'_1}}$.
    Thus, we conclude as in the previous case.
  \end{itemize}
  \item $\pos{r} = \pos{1r'}$. This case is symmetric to the one presented
  above. We have $t = \appterm{t_1}{t_2}$ and $t' = \appterm{t_1}{t'_2}$ with
  $\reduction{\pos{r'}}{t_2}{t'_2}{\beta}$, and $$\Phi_{t} = \Rule{\Phi_{t_1}
  \quad \many{\Phi_{t_2}^{i}}{i \in I}}{\sequ{\Gamma}{\assign{\appterm{t_1}{t_2}}{\tau}}}{}
  \rewrite{\beta} \Rule{\Phi_{t_1} \quad \many{\Phi_{t'_2}^{i}}{i \in
  I}}{\sequ{\Gamma}{\assign{\appterm{t_1}{t'_2}}{\tau}}}{} = \Phi_{t'}$$ with
  $\many{\Phi_{t_2}^{i} \rewrite{\beta} \Phi_{t'_2}^{i}}{i \in I}$. Then, if
  $\pos{p} = \pos{0q}$ (\ie $\pos{r} \not\prefix \pos{p}$) we may conclude
  by definition with $\residuals{\pos{p}}{\pos{r}} = \set{\pos{p}}$. Otherwise,
  $\pos{p} = \pos{1q}$ and $\pos{p} \in \toc{\Phi_{t}}$ iff $\pos{q} \in
  \toc{\Phi_{t'_2}^{k}}$ for some $k \in I$. Thus, the result follows by
  definition from the inductive hypothesis.
\end{itemize}
\vspace{-2em}
\end{proof}

\subsection{Weak-Head Needed Redexes are Typed}
\label{s:whn-typed}

In order to show that every weak-head needed redex occurrence corresponds to
a typed occurrence in some type derivation we start by proving that typed
occurrences do not come from untyped ones.

\begin{lemma}
\label{lem:typedAncestor}
Let $\reductionn{\rho}{\Phi_{t}}{\Phi_{t'}}{\beta}$ and $\pos{p} \in
\oc{t}$. If there exists $\pos{p'} \in \residuals{\pos{p}}{\rho}$ such that
$\pos{p'} \in \toc{\Phi_{t'}}$, then $\pos{p} \in \toc{\Phi_{t}}$.
\end{lemma}

\begin{proof}
Straightforward induction on $\rho$ using Lem.~\ref{lem:typedDescendant}.
\end{proof}

\begin{theorem}
\label{thm:whn-typed}
Let $\pos{r}$ be a weak-head needed redex in $t$. Let $\Phi$ be a type
derivation of  $t$. Then,  $\pos{r} \in \toc{\Phi}$.
\end{theorem}

\begin{proof}
By Thm.~\ref{thm:leftmost}, $\pos{r}$ is used in the weak-head reduction from
$t$ to $t' \in \WHNF{\beta}$. By Rem.~\ref{rem:typedLeftmost}, the weak-head
reduction contracts only typed redexes. Thus, $\pos{r}$ or some of its
residuals is a typed occurrence in its corresponding derivation tree. Finally,
we conclude by Lem.~\ref{lem:typedAncestor}, $\pos{r} \in \toc{\Phi}$.
\end{proof}

\subsection{Principally Typed Redexes are Weak-Head Needed}
\label{s:principal-whn}

As mentioned before, the converse of Thm.~\ref{thm:whn-typed} does not hold:
there are some typed occurrences that do not correspond to any weak-head needed
redex occurrence. This can be illustrated in the following examples (recall
$\Phi_{KI\Omega}$ defined in Sec.~\ref{sec:systemV}): $$
\Rule{
  \Phi_{KI\Omega}
  \vphantom{\Rule{}{\Gamma}{}}
}{
  \sequ{}{\assign{\absterm{y}{KI\Omega}}{\functtype{\intertype{}{}}{\valuetype}}}
}{\ruleArrowI}
\quad
\Rule{
  \Rule{}{
    \sequ{\assign{y}{\intertype{\functtype{\intertype{\valuetype}{}}{\valuetype}}{}}}{\assign{y}{\functtype{\intertype{\valuetype}{}}{\valuetype}}}
  }{\ruleAxiom}
  \quad
  \Phi_{KI\Omega}
}{
  \sequ{\assign{y}{\intertype{\functtype{\intertype{\valuetype}{}}{\valuetype}}{}}}{\assign{\appterm{y}{(KI\Omega)}}{\valuetype}}
}{\ruleArrowE} $$

Indeed, the occurrence $\pos{0}$ (resp $\pos{1}$) in the term
$\absterm{y}{KI\Omega}$ (resp. $\appterm{y}{(KI\Omega)}$) is typed but not
weak-head needed, since both terms are already in weak-head normal form.
Fortunately, typing relates to redex occurrences if we restrict type
derivations to \emph{principal} ones: given a term $t$ in weak-head
$\beta$-normal form, the derivation
$\derivable{\Phi}{\sequ{\Gamma}{\assign{t}{\tau}}}{\V}$ is
\emphdef{normal principally typed} if:
\begin{itemize}
  \item $t = \appterm{\appterm{\appterm{x}{t_1}}{\ldots}}{t_n}\ (n \geq 0)$, and
  $\Gamma = \set{\assign{x}{\intertype{\functtype{\overbrace{\functtype{\functtype{\intertype{}{}}{\ldots}}{\intertype{}{}}}^{\text{$n$ times}}}{\tau}}{}}}$
  and $\tau$ is a type variable  $\alpha$  (\ie none of the $t_i$ are typed), or 
  
  \item $t = \absterm{x}{t'}$, and $\Gamma = \varnothing$ and $\tau = \valuetype$. 
\end{itemize}

Given a weak-head normalising term $t$ such that
$\derivable{\Phi_t}{\sequ{\Gamma}{\assign{t}{\tau}}}{\V}$, we say that
$\Phi_{t}$ is \emphdef{principally typed} if $\Phi_{t} \rewriten{\beta}
\Phi_{t'}$ for some $t' \in \whnf{\beta}{t}$ implies $\Phi_{t'}$ is normal
principally typed.

Note in particular that the previous definition does not depend on the
chosen weak-head normal form $t'$: suppose $t'' \in \whnf{\beta}{t}$ is another
weak-head normal form of $t$, then $t'$ and $t''$ are convertible terms by the
Church-Rosser property~\cite{Barendregt84} so that $t'$ can be normal
principally typed iff $t''$ can, by Thm.~\ref{thm:subjectReductionV}
and~\ref{thm:subjectExpansionV}.

\begin{lemma}
\label{lem:typedRedex}
Let $\Phi_{t}$ be a type derivation with subject $t$ and $\pos{r} \in \roc{t}
\cap \toc{\Phi_{t}}$. Let $\reductionn{\rho}{\Phi_{t}}{\Phi_{t'}}{\beta}$ such
that $\Phi_{t'}$ is normal. Then, $\pos{r}$ is used in $\rho$.
\end{lemma}

\begin{proof}
Straightforward induction on $\rho$ using Lem.~\ref{lem:typedDescendant}.
\end{proof}

The notions of leftmost and weak-head needed reductions on (untyped) terms
naturally extends to \emph{typed} reductions on tree derivations. We thus have:

\begin{lemma}
\label{lem:typedLeftmost}
Let $t$ be a weak-head normalising term and $\Phi_{t}$ be principally typed.
Then, a leftmost typed reduction sequence starting at $\Phi_{t}$ is 
weak-head needed.
\end{lemma}

\begin{proof}
By induction on the leftmost typed sequence (called $\rho$). If $\rho$ is empty
the result is immediate. If not, we show that $t$ has a typed weak-head needed
redex (which is leftmost by definition) and conclude by inductive hypothesis.
Indeed, assume $t \in \WHNF{\beta}$. By definition $\Phi_{t}$ is normal
principally typed and thus it has no typed redexes. This contradicts $\rho$
being non-empty. Hence, $t$ has a weak-head redex $\pos{r}$ (\ie $t \notin
\WHNF{\beta}$). Moreover, $\pos{r}$ is both typed (by
Rem.~\ref{rem:typedLeftmost}) and weak-head needed (by
Lem.~\ref{lem:leftmostNeeded}). Thus, we conclude.
\end{proof}

\begin{theorem}
\label{thm:typed-whn}
Let $t$ be a weak-head normalising term, $\Phi_{t}$ be principally typed and
$\pos{r} \in \roc{t} \cap \toc{\Phi_{t}}$. Then, $\pos{r}$ is a weak-head
needed redex in $t$.
\end{theorem}

\begin{proof}
Let $\reductionn{\rho}{\Phi_{t}}{\Phi_{t'}}{\beta}$ be the leftmost typed
reduction sequence where $\Phi_{t'}$ is normal. Note that $\Phi_{t'}$ exists
by definition of \emph{principally typed}. By Lem.~\ref{lem:typedLeftmost},
$\rho$ is a weak-head needed reduction sequence. Moreover, by
Lem.~\ref{lem:typedRedex}, $\pos{r}$ is used in $\rho$. Hence, $\pos{r}$ is a
weak-head needed redex in $t$.
\end{proof}

As a direct consequence of Thm.~\ref{thm:whn-typed} and~\ref{thm:typed-whn},
given a weak-head normalising term $t$, the typed redex occurrences in its
principally typed derivation (which always exists) correspond to its weak-head
needed redexes. Hence, system $\V$ allows to identify all the weak-head needed
redexes of a weak-head normalising term.



\section{Characterising Weak-Head Needed Normalisation}
\label{sec:charact}

This section presents one of the main pieces contributing to our observational
equivalence  result. Indeed, we relate typing with weak-head neededness by
showing that any typable term in system $\V$ is normalising for weak-head
needed reduction. This characterisation highlights the power of intersection
types. We start by a technical lemma.

\begin{lemma}
\label{lem:PhiNForm-WHNF}
Let $\derivable{\Phi}{\sequ{\Gamma}{\assign{t}{\tau}}}{\V}$. Then, $\Phi$
normal implies $t \in \WHNF{\beta}$.
\end{lemma}

\begin{proof}
By induction on $\Phi$ analysing the last rule applied.
\begin{itemize}
  \item $\ruleAxiom$. Then $t = x \in \WHNF{\beta}$.
  
  \item $\ruleValue$. Then $t = \absterm{x}{t'} \in \WHNF{\beta}$.
  
  \item $\ruleArrowI$. Then $t = \absterm{x}{t'} \in \WHNF{\beta}$.
  
  \item $\ruleArrowE$. Then $t = \appterm{r}{u}$ with $\Gamma =
  \ctxtsum{\Gamma'}{\Delta_i}{i \in I}$,
  $\derivable{\Phi'}{\sequ{\Gamma'}{\assign{r}{\functtype{\intertype{\sigma_i}{i
  \in I}}{\tau}}}}{\V}$ and $(\sequ{\Delta_i}{\assign{u}{\sigma_i}})_{i
  \in I}$ for some $I$. From $\Phi$ normal we have $\Phi'$ normal too. Thus, by
  inductive hypothesis, $r \in \WHNF{\beta}$. Moreover, $r \neq \absterm{y}{r'}$
  since $\rootpos \in \toc{\Phi}$. Then, $r =
  \appterm{\appterm{\appterm{x}{r_1}}{\ldots}}{r_n}$ and we conclude with $t =
  \appterm{\appterm{\appterm{\appterm{x}{r_1}}{\ldots}}{r_n}}{u} \in
  \WHNF{\beta}$.
\end{itemize}
\vspace{-2em}
\end{proof}

Let $\reductionn{\rho}{t_1}{t_n}{\beta}$. We say that $\rho$ is a
\emphdef{left-to-right} reduction sequence iff for every $i < n$ if
$\reduction{\pos{r_i}}{t_i}{t_{i+1}}{\beta}$ and $\pos{l_i}$ is to the
left of $\pos{r_i}$ then, for every $j > i$ such that
$\reduction{\pos{r_j}}{t_j}{t_{j+1}}{\beta}$ we have that $\pos{r_j} \notin
\residuals{\set{\pos{l_i}}}{\rho_{ij}}$ where $\reductionn{\rho_{ij}}{t_i}{t_j}{\beta}$
is the corresponding subsequence of $\rho$. In other words, for every $j$ and every
$i < j$, $\pos{r_j}$ is not a residual of a redex to the left of $\pos{r_i}$
(relative to the given reduction subsequence from $t_i$ to
$t_j$)~\cite{Barendregt84}.

Left-to-right reductions define in particular standard strategies, which 
give canonical ways to construct reduction sequences from one term to another:

\begin{theorem}[\cite{Barendregt84}]
\label{thm:standard}
If $t \rewriten{\beta} t'$, there exists a left-to-right reduction from $t$ to $t'$.
\end{theorem}

\begin{theorem}
\label{thm:typedIffNeeded}
Let $t \in \Term$. Then,
$\derivable{\Phi}{\sequ{\Gamma}{\assign{t}{\tau}}}{\V}$ iff $t \in
\WN{\weaknd}$.
\end{theorem}

\begin{proof}
$\Rightarrow$) By Thm.~\ref{thm:subjectReductionV} we know that the strategy
reducing only typed redex occurrences is normalising, \ie there exist $t'$ and
$\Phi'$ such that $t \rewriten{\beta} t'$,
$\derivable{\Phi'}{\sequ{\Gamma}{\assign{t'}{\tau}}}{\V}$ and $\Phi'$ normal.
Then, by Lem.~\ref{lem:PhiNForm-WHNF}, $t' \in \WHNF{\beta}$. By
Thm.~\ref{thm:standard}, there exists a left-to-right reduction
$\reductionn{\rho}{t}{t'}{\beta}$. Let us write $$\rho : t = t_1
\rewriten{\beta} t_n \rewriten{\beta} t'$$ such that $t_1, \ldots, t_{n-1}
\notin \WHNF{\beta}$ and $t_n \in \WHNF{\beta}$.

We claim that all reduction steps in $t_1 \rewriten{\beta} t_n$ are leftmost.
Assume towards a contradiction that there exists $k < n$ such that
$\reduction{\pos{r}}{t_k}{t_{k+1}}{\beta}$ and $\pos{r}$ is not the leftmost
redex of $t_k$ (written $\pos{l_{k}}$). Since $\rho$ is a left-to-right
reduction, no residual of $\pos{l_{k}}$ is contracted after the $k$-th step.
Thus, there is a reduction sequence from $t_k \notin \WHNF{\beta}$ to $t_n \in
\WHNF{\beta}$ such that $\pos{l_{k}}$ is not used in it. This leads to a
contradiction with $\pos{l_{k}}$ being weak-head needed in $t_k$ by
Lem.~\ref{lem:leftmostNeeded}.

As a consequence, there is a leftmost reduction sequence $t \rewriten{\beta}
t_n$. Moreover, by Lem.~\ref{lem:leftmostNeeded}, $t \rewriten{\weaknd} t_n
\in \WHNF{\beta} = \NF{\weaknd}$. Thus, $t \in \WN{\weaknd}$.

$\Leftarrow$) Consider the reduction $\reductionn{\rho}{t}{t'}{\weaknd}$ with
$t' \in \whnf{\beta}{t}$. Let
$\derivable{\Phi'}{\sequ{\Gamma}{\assign{t'}{\tau}}}{\V}$ be the normal
principally typed derivation for $t'$ as defined in Sec.~\ref{s:principal-whn}.
Finally, we conclude by induction in $\rho$ using
Thm.~\ref{thm:subjectExpansionV},
$\derivable{\Phi}{\sequ{\Gamma}{\assign{t}{\tau}}}{\V}$.
\end{proof}

\section{The Call-by-Need Lambda-Calculus}
\label{sec:call-by-need}

This section describes the syntax and the operational semantics of the
call-by-need lambda-calculus introduced in~\cite{AccattoliBM14}. It is more
concise than previous specifications of
call-by-need~\cite{AriolaFMOW95,AriolaF97,MaraistOW98,ChangF12}, but it is
operationally equivalent to them~\cite{BalabonskiBBK17}, so that our results
could also be presented by using alternative specifications.

Given a countable infinite set $\TermVariable$ of variables $x, y, z, \ldots$
we define different syntactic categories for terms, values, list contexts,
answers and need contexts: 

\begin{center}
\begin{tabular}{rrcll}
\textbf{(Terms)}         & $t, u$               & $\Coloneq$ & $x \in \TermVariable \mid \appterm{t}{u} \mid \absterm{x}{t} \mid \substerm{x}{u}{t}$ \\
\textbf{(Values)}        & $v$                  & $\Coloneq$ & $\absterm{x}{t}$ \\
\textbf{(List contexts)} & $\ctxt{L}$           & $\Coloneq$ & $\Box \mid \substerm{x}{t}{\ctxt{L}}$ \\ 
\textbf{(Answers)}       & $a$                  & $\Coloneq$ & $\ctxtapp{\ctxt{L}}{\absterm{y}{t}}$ \\
\textbf{(Need contexts)} & $\ctxt{M}, \ctxt{N}$ & $\Coloneq$ & $\Box \mid \appterm{\ctxt{N}}{t} \mid \substerm{x}{t}{\ctxt{N}} \mid \substerm{x}{\ctxt{M}}{\ctxtwoc{\ctxt{N}}{x}}$
\end{tabular}
\end{center}

We denote the set of terms by $\TermExplicit$. Terms of the form
$\substerm{x}{u}{t}$ are \emphdef{closures}, and $\exsubs{x}{u}$
is called an \emphdef{explicit substitution} (ES). The set of
$\TermExplicit$-terms without ES is the set of  \emph{terms of the
$\lambda$-calculus}, \ie $\Term$. The notions of \emphdef{free}
and \emphdef{bound} variables are defined as expected, in
particular, $\fv{\substerm{x}{u}{t}} \eqdef \fv{t} \setminus \set{x} \cup \fv{u}$, 
$\fv{\absterm{x}{t}} \eqdef \fv{t} \setminus \set{x}$,  
$\bv{\substerm{x}{u}{t}} \eqdef \bv{t} \cup \set{x} \cup \bv{u}$
and $\bv{\absterm{x}{t}} \eqdef \bv{t} \cup \set{x}$. We extend the
standard notion of \emphdef{$\alpha$-conversion} to ES,
as expected.

We use the special notation $\ctxtwoc{\ctxt{N}}{u}$ or $\ctxtwoc{\ctxt{L}}{u}$
when the free variables of $u$ are not captured by the context, \ie
there are no abstractions or explicit substitutions in the context
that binds the free variables of $u$. Thus for example, given $\ctxt{N} =
\substerm{x}{z}{(\appterm{\Box}{x})}$, we have $\substerm{x}{z}{(\appterm{y}{x})}
= \ctxtapp{\ctxt{N}}{y} = \ctxtwoc{\ctxt{N}}{y}$, but
$\substerm{x}{z}{(\appterm{x}{x})} = \ctxtapp{\ctxt{N}}{x}$ cannot be
written as $\ctxtwoc{\ctxt{N}}{x}$. Notice the use of this special
notation in the last case of needed contexts, an example of such case
being $\substerm{x}{\Box}{\substerm{y}{t}{(\appterm{x}{y})}}$.

\bigskip
The \emphdef{call-by-need calculus}, introduced in~\cite{AccattoliBM14}, is
given by the set of terms $\TermExplicit$ and the \emphdef{reduction
relation} $\rewrite{\callbyneed}$, the \emphdef{union} of $\rewrite{\dB}$
and $\rewrite{\lsv}$, which are, respectively, the closure by \emph{need
contexts} of the following rewriting rules: $$
\begin{array}{rcl}
\appterm{\ctxtapp{\ctxt{L}}{\absterm{x}{t}}}{u}            & \rrule{\dB}  & \ctxtapp{\ctxt{L}}{\substerm{x}{u}{t}} \\
\substerm{x}{\ctxtapp{\ctxt{L}}{v}}{\ctxtwoc{\ctxt{N}}{x}} & \rrule{\lsv} & \ctxtapp{\ctxt{L}}{\substerm{x}{v}{\ctxtwoc{\ctxt{N}}{v}}}
\end{array} $$
These rules avoid capture of free variables. An example of
$\callbyneed$-reduction sequence is the following, where the redex of each step
is underlined for clearness: $$
\begin{array}{l@{\ }l@{\ }l@{\ }l}
\underline{\appterm{(\absterm{x_1}{\appterm{I}{(\appterm{x_1}{I})}})}{(\absterm{y}{\appterm{I}{y}})}}                        & \rewrite{\dB}  &
\substerm{x_1}{\absterm{y}{\appterm{I}{y}}}{(\underline{\appterm{I}{(\appterm{x_1}{I})}})}                                   & \rewrite{\dB}  \\
\underline{\substerm{x_1}{\absterm{y}{\appterm{I}{y}}}{\substerm{x_2}{\appterm{x_1}{I}}{x_2}}}                               & \rewrite{\lsv} &
\substerm{x_1}{\absterm{y}{\appterm{I}{y}}}{\substerm{x_2}{\underline{\appterm{(\absterm{x_3}{\appterm{I}{x_3}})}{I}}}{x_2}} & \rewrite{\dB}  \\
\substerm{x_1}{\absterm{y}{\appterm{I}{y}}}{\substerm{x_2}{\substerm{x_3}{I}{(\underline{\appterm{I}{x_3}})}}{x_2}}          & \rewrite{\dB}  &
\substerm{x_1}{\absterm{y}{\appterm{I}{y}}}{\substerm{x_2}{\underline{\substerm{x_3}{I}{\substerm{x_4}{x_3}{x_4}}}}{x_2}}    & \rewrite{\lsv} \\
\substerm{x_1}{\absterm{y}{\appterm{I}{y}}}{\substerm{x_2}{\substerm{x_3}{I}{\underline{\substerm{x_4}{I}{x_4}}}}{x_2}}      & \rewrite{\lsv} &
\substerm{x_1}{\absterm{y}{\appterm{I}{y}}}{\underline{\substerm{x_2}{\substerm{x_3}{I}{\substerm{x_4}{I}{I}}}{x_2}}}        & \rewrite{\lsv} \\
\substerm{x_1}{\absterm{y}{\appterm{I}{y}}}{\substerm{x_3}{I}{\substerm{x_4}{I}{\substerm{x_2}{I}{I}}}}
\end{array} $$

As for call-by-name, reduction preserves free variables, \ie
$t \rewrite{\callbyneed} t'$ implies $\fv{t} \supseteq \fv{t'}$.
Notice that call-by-need reduction is also weak, so that
answers are not $\callbyneed$-reducible.


\section{Observational Equivalence}
\label{sec:results}

The results in Sec.~\ref{sec:charact} are used here to prove soundness and
completeness of call-by-need w.r.t weak-head neededness, our second main
result. More precisely, a call-by-need interpreter stops in a value if and only
if the weak-head needed reduction stops in a value. This means that
call-by-need and call-by-name are observationally equivalent.

Formally, given a reduction relation $\R$ on a term language $\mathcal{T}$, and
an associated notion of context for $\mathcal{T}$, we define $t$ to be 
\emphdef{observationally equivalent} to $u$, written $t \obs{\R} u$, iff
$\ctxtapp{\ctxt{C}}{t} \in \WN{\R} \Leftrightarrow \ctxtapp{\ctxt{C}}{u} \in
\WN{\R}$ for every context $\ctxt{C}$. In order to show our final result we
resort to the following theorem:

\begin{theorem}[\cite{Kesner16}]
\label{thm:name-need}
\begin{enumerate}
\item \label{uno} 
Let $t \in \Term$. Then, $\derivable{\Phi}{\sequ{\Gamma}{\assign{t}{\tau}}}{\V}$ iff $t \in \WN{\callbyname}$.
\item \label{dos}
  For all  terms $t$ and $u$ in $\Term$, $t \obs{\callbyname} u$ iff
  $t \obs{\callbyneed} u$.
   \end{enumerate}
\end{theorem}

These observations allows us to conclude:

\begin{theorem}
For all  terms $t$ and $u$ in $\Term$, $t \obs{\weaknd} u$ iff
$t \obs{\callbyneed} u$. 
\end{theorem}

\begin{proof} By Thm.~\ref{thm:name-need}:\ref{dos} it is sufficient to show 
$t \obs{\weaknd} u$ iff $t \obs{\callbyname} u$. The proof proceeds as follows: $$
\vspace{-1.5em}
\begin{array}{rclc@{\qquad}l}
t                                            & \obs{\callbyname} & u                                            & \text{iff} & \text{(definition)} \\
\ctxtapp{\ctxt{C}}{t} \in \WN{\callbyname}   & \Leftrightarrow   & \ctxtapp{\ctxt{C}}{u} \in \WN{\callbyname}   & \text{iff} & \text{(Thm.~\ref{thm:name-need}:\ref{uno})} \\
\ctxtapp{\ctxt{C}}{t} \text{ typable in } \V & \Leftrightarrow   & \ctxtapp{\ctxt{C}}{u} \text{ typable in } \V & \text{iff} & \text{(Thm.~\ref{thm:typedIffNeeded})} \\
\ctxtapp{\ctxt{C}}{t} \in \WN{\weaknd}       & \Leftrightarrow   & \ctxtapp{\ctxt{C}}{u} \in \WN{\weaknd}       & \text{iff} & \text{(definition)} \\
t                                            & \obs{\weaknd}     & u
\end{array} $$
\end{proof}


\section{Conclusion}
\label{sec:conclusion}

We establish a clear connection between the semantical standard notion of
neededness and the syntactical concept of call-by-need. The use of
non-idempotent types --a powerful technique being able to characterise
different operational properties-- provides a simple and natural tool to show
observational equivalence between these two notions. 
We refer the reader to~\cite{BalabonskiTh} for other proof techniques (not
based on intersection types) used to connect semantical notions of neededness
with syntactical notions of lazy evaluation.

An interesting (and not difficult) extension of our result in
Sec.~\ref{sec:redexes} is that call-by-need reduction (defined on
$\lambda$-terms with explicit substitutions) contracts only $\dB$ weak-head
needed redexes, for an appropriate (and very natural) notion of weak-head
needed redex for $\lambda$-terms with explicit substitutions. A technical tool
to obtain such a result would be the type system $\A$~\cite{Kesner16}, a
straightforward adaptation of system $\V$ to call-by-need syntax. 

Given the recent formulation of \emph{strong
call-by-need}~\cite{BalabonskiBBK17} describing a deterministic call-by-need
strategy to normal form (instead of weak-head normal form), it would be natural
to extend our technique to obtain an observational equivalence result between
the standard notion of needed reduction (to full normal forms) and the strong
call-by-need strategy. This remains as future work.

\bibliographystyle{plain}
\bibliography{biblio}

\end{document}